\documentclass[12pt]{article}
\usepackage[authoryear,round]{natbib}
\usepackage{graphicx}
\usepackage{latexsym}
\usepackage[left=1.3in,top=1.3in,right=1.3in,bottom=1.3in]{geometry}\usepackage{amsmath}
\usepackage{booktabs}
\usepackage{color}
\usepackage{tikz}
\usepackage{bbm,bm}
\usepackage{amsmath,amsfonts,amsthm}

\newtheorem{theorem}{Theorem}
\newtheorem{lemma}{Lemma}

\newcommand{\pl}{\parallel}
\newcommand{\openr}{\hbox{${\rm I\kern-.2em R}$}}
\newcommand{\openn}{\hbox{${\rm I\kern-.2em N}$}}

\DeclareMathOperator{\E}{\mbox{E}}
\newcommand{\Dstarn}{D^*(\cdot \mid \bar{Q}_{n}^*, Q_{n,W}, \bar{G}_n)}
\newcommand{\Dstarnadapt}{D^*(\cdot \mid \bar{Q}_{n}^\#, Q_{n,W}, \bar{G}_n(\cdot \mid \bar{Q}_n))}
\newcommand{\Dstarnadaptinit}{D^*(\cdot \mid \bar{Q}_{n}, Q_{n,W}, \bar{G}_n(\cdot \mid \bar{Q}_n))}

\newcommand{\Dstarzeroadapt}{D^*(\cdot \mid \bar{Q}_{0}, Q_{0,W}, \bar{G}_0(\cdot \mid \bar{Q}_0))}
\newcommand{\Dstarlim}{D^*(\cdot \mid \bar{Q}^*, Q_{0,W}, \bar{G})}
\newcommand{\Remstar}{R_2(\bar{Q}^*_n, \bar{Q}_0, \bar{G}_n, \bar{G}_0, Q_{0,W})}
\newcommand{\Remstaradapt}{R_2(\bar{Q}^\#_n, \bar{Q}_0, \bar{G}_n(\cdot \mid \bar{Q}_n), \bar{G}_0, Q_{0,W})}
\newcommand{\Remstaradaptinit}{R_2(\bar{Q}_n, \bar{Q}_0, \bar{G}_n(\cdot \mid \bar{Q}_n), \bar{G}_0, Q_{0,W})}
\DeclareMathOperator{\pr}{\mbox{pr}}
\DeclareMathOperator{\ind}{\mathbbm{1}}

\definecolor{mypurple}{RGB}{75,46,151}
\definecolor{mypurple_dark}{RGB}{56.25,34.50,113.25}
\definecolor{mypurple_darker}{RGB}{37.5,23,75.5}
\definecolor{mypurple_light}{RGB}{100,61.3333,201.33333}

\usepackage{listings}
\definecolor{codeback}{RGB}{244,243,243}

\title{A nonparametric super-efficient estimator of the average treatment effect}
\author{David Benkeser \\
Department of Biostatistics and Bioinformatics, Emory University \\ \vspace{-0.15in} \\
Wilson Cai and Mark J. van der Laan   \\ 
Division of Biostatistics, UC Berkeley\\
 }

\begin{document}
\maketitle

\begin{abstract} 

Doubly robust estimators of causal effects are a popular means of estimating causal effects. Such estimators combine an estimate of the conditional mean of the outcome given treatment and confounders (the so-called outcome regression) with an estimate of the conditional probability of treatment given confounders (the propensity score) to generate an estimate of the effect of interest. In addition to enjoying the double-robustness property, these estimators have additional benefits. First, flexible regression tools, such as those developed in the field of machine learning, can be utilized to estimate the relevant regressions, while the estimators of the treatment effects retain desirable statistical properties. Furthermore, these estimators are often statistically efficient, achieving the lower bound on the variance of regular, asymptotically linear estimators. However, in spite of their asymptotic optimality, in problems where causal estimands are weakly identifiable, these estimators may behave erratically. We propose two new estimation techniques for use in these challenging settings. Our estimators build on two existing frameworks for efficient estimation: targeted minimum loss estimation and one-step estimation. However, rather than using an estimate of the propensity score in their construction, we instead opt for an alternative regression quantity when building our estimators: the conditional probability of treatment given the conditional mean outcome. We discuss the theoretical implications and demonstrate the estimators' performance in simulated and real data. 

\end{abstract}
{\bf Key words:} causal inference, average treatment effect, asymptotic linearity, efficient influence function, collaborative targeted minimum loss estimation, super efficiency

\section{Introduction}\label{section1}

In many areas of research, the scientific question of interest is often answered by drawing statistical inference about the average effect of a treatment on an outcome. Depending on the setting, this ``treatment'' might correspond to an actual therapeutic treatment, a harmful exposure, or a policy intervention. We use $Y(1)$ to denote the potential outcome of a typical data unit sampled from the population of interest when the unit receives the treatment of interest, and $Y(0)$ to denote the potential outcome if that unit instead receives control. In this work, we focus on estimation of the average treatment effect (ATE), the average difference between $Y(1)$ and $Y(0)$ in the population of interest. 

Often, due to ethical or logistical constraints, we cannot randomly assign data units to receive/not receive the treatment. Thus, in order to draw valid conclusions about the ATE, we require statistical methodology that can control for confounders of the treatment/outcome relationship. Although epidemiological and other applied literatures are still considering the relative merits of various methodologies, the statistical literature has provided direction through consideration of semiparametric efficient methods. The literature provides many examples of such estimators. In some situations, regularized or sieve maximum likelihood estimators can be used (e.g., \citet{Vaart98}); however, this generally requires careful selection of tuning parameters. Existing literature lacks general guidelines for how such parameters can be chosen in practice, which limits the utility of these strategies. On the other hand, methods that are built around a causal effect parameter's efficient influence function offer a more straightforward approach to estimation. Foremost amongst these approaches are one-step estimation \citep{ibragimov1981,pfanzagl1982,Bickeletal97} and targeted minimum loss estimation (TMLE) \citep{vanderLaan&Rubin06,vanderLaan:Rose11}. 

The efficient influence function often depends on the observed data distribution through certain key nuisance parameters. In the context of the ATE, these nuisance parameters are the conditional mean of the outcome given treatment and confounders (the so-called outcome regression, OR), the conditional probability of treatment given confounders (the so-called propensity score, PS), and the distribution function of confounders in the population of interest. Once estimators of these key quantities are available, each methodology provides its own recipe for combining the relevant nuisance estimators into an estimate of the causal effect of interest. Assuming the nuisance estimators satisfy certain regularity conditions, the resulting estimators of the ATE, when appropriately centered and scaled, converge in distribution to a mean-zero Gaussian variate with variance equal to the semiparametric Cramer-Rao efficiency bound for regular estimators. In addition to efficiency, these estimators are also doubly-robust, meaning that they are consistent for the causal effect of interest even if one of the OR or PS is inconsistently estimated. 

One of the key assumptions underlying any methodology for estimating the ATE using observational data is the strong positivity assumption, which stipulates that the PS must be bounded between zero and one almost everywhere. That is, if we define strata of data units based on their observed confounders, any stratum with positive support must have some probability of receiving and not receiving the treatment. If this condition fails, then the ATE is not estimable from observed data. Moreover, even if the condition holds theoretically, in practice there may be small estimated propensity scores -- so-called \emph{practical} violations of the positivity assumption \citep{Petersenetal12}. In such cases, the one-step estimator and TMLE can suffer from erractic, non-robust behavior. 

To combat this behavior, various extensions have been proposed including collaborative TMLE (CTMLE) \citep{vanderLaan:Gruber10,Gruber&vanderLaan10a,Stitelman:vanderLaan10,Wang:Rose:vanderLaan11}. In this approach the OR and PS are estimated \emph{collaboratively}, by selecting an estimate of the PS based on how well it tunes an estimate of the OR. From a theoretical point of view, the goal of CTMLE is generally to provide an estimator that is more robust than TMLE, but that nevertheless maintains TMLE's asymptotic efficiency. CTMLE enjoys additional theoretical benefits in settings where the OR is inconsistently estimated. In this case, the CTMLE enjoys a more general robustness property than TMLE referred to as collaborative double-robustness \citep{vanderLaan:Gruber10}. Many proposed CTMLEs are designed specifically for settings with practical positivity violations, for example by: choosing a truncation level for estimated propensities, selecting variables to be included in the PS, or tuning particular machine learning algorithms \citep{Juetal17,Juetal17a,Juetal18}. These works show that CTMLE can provide greater robustness than TMLE in challenging situations. 

In spite of the putative benefits of CTMLE, these estimators are not widely used. There may be several reasons why. First, the approach involves many decision points for the analyst, who must select an increasingly complex sequence of estimators for the PS and implement each of these estimators. Second, the approach often involves extended computation time relative to traditional doubly-robust methods. In particular, cross-validation is needed to validate which PS method should be selected from amongst the user-chosen sequence. Third, from a theoretical perspective, performing robust inference based on existing CTMLE approaches is also challenging, involving either strong assumptions on nuisance estimators (see Appendix 17 of \citet{vanderLaan:Rose11}) or additional iterative computational steps \citep{vanderLaan14a}. 

In this work, we seek to overcome these limitations by proposing a new approach to CTMLE. The twist in the present proposal relative to existing CTMLE approaches is that we assume the OR estimator is consistent at a fast rate. 
Assuming this consistency, we provide an alternative target for a PS estimator that is explicitly adaptive to the OR. In particular, rather than estimating the true PS, we recommend estimating the propensity for receiving treatment as a function of the estimated OR. This low-dimensional regression can be substituted in place of an estimator of the true PS in a standard TMLE or one-step procedure. We show that, when appropriately scaled, the resultant estimator is asymptotically Normal with variance that is generally smaller than that of an efficient estimator. Thus, our proposal provides a new approach to CTMLE that is tailored both for small- and large-sample and performance.


\section{Background}
\subsection{Identification of ATE}

Suppose we observe $n$ independent copies of the data unit $O := (W, A, Y)$, where $W \in \mathcal{W}$ is a vector of putative confounders, $A \in \{0,1\}$ is a binary treatment, and $Y \in [0,1]$ is the outcome of interest. Our assumption that $Y \in [0,1]$ does not sacrifice any generality of our proposed methodology. We denote by $P_0$ the probability distribution of $O$ and assume that $P_0$ is an element of a statistical model $\mathcal{M}$. We take $\mathcal{M}$ to be a nonparametric model. 

As above, we use $Y(1)$ and $Y(0)$ to denote the counterfactual outcomes under treatment and no treatment, respectively. For $a = 0,1$, we denote by $P_0^a$ the probability distribution of $Y(a)$. The ATE is defined as $\E_{P_0^1}[Y(1)] - \E_{P_0^0}[Y(0)]$, the difference in average outcome if the entire population were assigned to receive $A = 1$ versus $A = 0$. The ATE is identifiable from the observed data under the following assumptions: consistency, $Y = Y(a) \mid A = a$; no interference: $Y_i(a)$ does not depend on $A_j$ for $j \ne i$; ignorability: $A \perp (Y(1),Y(0)) | W$; and positivity: $\pr_{P_0}\{ 0 < \mbox{pr}_{P_0}(A = 1 \mid W) < 1\} = 1$. The first two assumptions are needed in order to have well-defined counterfactual random variables, while the ignorability condition essentially states that there are no unmeasured confounders of $A$ and $Y$. As mentioned in the introduction, the positivity criterion stipulates that every unit has a chance of receiving $A = 1$ and $A = 0$. If these assumptions hold, the average treatment effect is identified by the G-computation formula \begin{equation} \label{gcomp}
    \E_{P_0^1}[Y^{(1)}] - E_{P_0^0}[Y^{(0)}] = \E_{P_0}[\E_{P_0}(Y \mid A = 1, W) - \E_{P_0}(Y \mid A = 0, W)] \ .
\end{equation} 

\subsection{Estimators of the ATE}

For simplicity, we hence consider estimation of $\psi_0 := \E_{P_0}[\E_{P_0}(Y \mid A = 1, W)]$, which we refer to as the treatment-specific mean. Symmetric arguments can be made about $\E_{P_0}[\E_{P_0}(Y \mid A = 0, W)]$, and thus the ATE. Hereafter, when we refer to the OR, it is understood that we are referring to the quantity $\E_{P_0}(Y \mid A = 1, W)$, the regression of $Y$ on $W$ amongst units observed to receive the treatment. 

For each $w \in \mathcal{W}$, we denote by $\bar{Q}_0(w) := \E_{P_0}(Y \mid A = 1, W = w)$ the true OR evaluated at $W = w$ and denote by $\bar{Q}_n(w)$ an estimate of $\bar{Q}_0(w)$ based on $O_1,\dots,O_n$. We use $\bar{\mathcal{Q}}$ to denote the model for the OR implied by $\mathcal{M}$; that is, $\mathcal{Q}$ is a collection of all possible ORs allowed by the model $\mathcal{M}$. Similarly, for each $w \in \mathcal{W}$, we denote by $\bar{G}_0(w)$ the true PS evaluated at $W = w$, by $\bar{G}_n(w)$ an estimate of $\bar{G}_0(w)$, and by $\bar{\mathcal{G}}$ the model for the PS implied by $\mathcal{M}$. Finally, for each $w \in \mathcal{W}$, we denote by $Q_{0,W}(w) := \mbox{pr}_{P_0}(W \le w)$ the distribution function of the vector of confounders. In the remainder, we use the empirical distribution $Q_{n,W}(w) := n^{-1}\sum_{i=1}^n \ind(W_i \le w)$, where $\ind$ is the indicator function, as estimate of $Q_{0,W}$. We denote by $\mathcal{Q}_W$ the model for $Q_{0,W}$ implied by $\mathcal{M}$. It is useful to our discussion to regard the parameter of interest as a mapping $\Psi$ from $\mathcal{Q} := \mathcal{\bar{Q}} \times \mathcal{Q}_W$ to $[0,1]$. That is, given a $Q := (\bar{Q}, Q_W) \in \mathcal{Q}$, $\Psi(Q) := \int \bar{Q}(u) dQ_W(u)$ is the value of the treatment-specific mean implied by the OR $\bar{Q}$ and confounder distribution $Q_W$. Thus, denoting by $Q_0 := (\bar{Q}_0, Q_{0,W})$ the true values of these quantities, we have $\psi_0 = \Psi(Q_0)$. 

We remind readers that a regular (see Appendix A) estimator $\psi_n$ of $\psi_0$ is \emph{asymptotically linear} if and only if $\psi_n - \psi_0$ behaves approximately as an empirical mean of a mean-zero, finite-variance function of the observed data. This function is referred to as the estimator's \emph{influence function}. Depending on the chosen model, there may be a large class of possible influence functions of regular estimators. The influence function in this class that has the smallest variance is referred to as the \emph{efficient influence function} (EIF). Any estimator with influence function equal to the EIF is said to be \emph{efficient} amongst regular, asymptotically linear estimators. Given $Q \in \mathcal{Q}$, $\bar{G} \in \mathcal{\bar{G}}$, and a typical observation $o$, we define the efficient influence function for the treatment-specific mean relative to $\mathcal{M}$, \[
	D^*(o \mid \bar{Q}, Q_W, \bar{G}) := \frac{a}{\bar{G}(w)} \left[y - \bar{Q}(w)\right] + \bar{Q}(w) - \int \bar{Q}(u) dQ_W(u) \ . 
\]

Several frameworks exist for constructing estimators with a user-specified influence function. By selecting the EIF, these frameworks can be used to generate efficient estimators. We focus on the one-step and targeted minimum loss estimation frameworks. For our discussion, it is useful to introduce the idea of a \emph{plug-in estimator}. We denote by $Q_n := (\bar{Q}_n, Q_{n,W})$ an estimate of $Q_0 := (\bar{Q}_0, Q_{0,W})$. A plug-in estimate has the form $\Psi(Q_n) = \int \bar{Q}_n(u) dQ_{n,W}(u) = n^{-1} \sum_{i=1}^n \bar{Q}_n(W_i)$. The one-step estimator $\psi_n^+$ applies an EIF-based correction to the plug-in estimate, \[
  \psi_{n}^+ := \Psi(Q_n) + \frac{1}{n}\sum_{i=1}^n D^*(O_i \mid \bar{Q}_n, Q_{n,W}, \bar{G}_n) \ . 
\]

TMLE is a general framework for constructing plug-in estimators that satisfy, possibly several, user-specified equations. The framework is implemented in two steps. First, initial estimators of relevant nuisance parameters (e.g., the OR and PS) are generated using a user-chosen technique. Subsequently, the initial estimates are carefully modified such that (i) the modified estimators inherit desirable properties of the initial estimators (e.g., their rate of convergence); and (ii) relevant, user-specified equations are satisfied. For the present problem, a TMLE can be implemented by first generating an initial estimate $\bar{Q}_n$ of the OR and $\bar{G}_n$ of the PS. The OR regression estimator is subsequently updated to a \emph{targeted} estimator $\bar{Q}_n^*$, such that the EIF estimating equation, $n^{-1} \sum_{i=1}^n D(O_i \mid \bar{Q}_n^*, Q_{n,W}, \bar{G}_n) = 0$, is satisfied. This can be achieved, for example, by defining a logistic regression working model for the OR with $\mbox{logit}(\bar{Q}_n)$ as an offset, no intercept term, and a single covariate $H_n$. For each $a \in \{0,1\}$ and $w \in \mathcal{W}$, we define this covariate as $H_n(a,w) := a/\bar{G}_n(w)$. The maximum likelihood estimator (MLE) $\epsilon_n$ of the regression coefficient $\epsilon$ associated with the covariate $H_n$ is estimated (e.g., via iteratively re-weighted least squares). For each $w \in \mathcal{W}$, we define the so-called \emph{targeted} OR estimator, $\bar{Q}_n^*(w) = \mbox{expit}\{\mbox{logit}[\bar{Q}_n(w)] + \epsilon_n H_n(1,w) \}$. It is straightforward to show that the score of the coefficient $\epsilon$ at $\epsilon = 0$ evaluated at a typical observation $o$, equals $D(o \mid \bar{Q}_n^*, \bar{G}_n, Q_n)$; thus, we may deduce that the EIF estimating equation is satisfied by the updated OR estimate $\bar{Q}_n^*$. The TMLE $\psi_n^*$ of the treatment-specific mean $\psi_0$ is computed as the plug-in estimator based on the modified OR estimator, $\psi_n^* = \int \bar{Q}_n^*(u) dQ_{n,W}(u) = n^{-1}\sum_{i=1}^n \bar{Q}_n^*(W_i)$. 

We remark that both the one-step and TMLE frameworks can be seen as first generating an initial estimate based on the OR and subsequently applying a correction procedure that involves an estimate of the PS. This view of the estimators is useful to our discussion below. 

\subsection{Large-sample theory and small-sample considerations}

We hence focus discussion on TMLE, with the understanding that similar arguments apply to the one-step estimator. The following equation (derived with assumptions in Appendix B) is useful for studying large-sample behavior of a TMLE, \begin{equation} \label{linearization}
	\Psi(Q_n^*) - \Psi(Q_0) = \frac{1}{n}\sum\limits_{i=1}^n D(O_i \mid \bar{Q}^*, Q_{0,W}, \bar{G}) + R_2(\bar{Q}_n, \bar{Q}_0, \bar{G}_n, \bar{G}_0, Q_{0,W}) + o_{\text{p}}(n^{-1/2}) \ , 
\end{equation}
where $(\bar{Q}^*, \bar{G})$ is the in-probability limit of $(\bar{Q}_n^*, \bar{G}_n)$. The term $R_2$ is a second-order remainder that involves a difference between $(\bar{Q}_n^*,\bar{G}_n)$ and $(\bar{Q}_0,\bar{G}_0)$, \begin{equation} \label{remainder}
	R_2(\bar{Q}_n, \bar{Q}_0, \bar{G}_n, \bar{G}_0, Q_{0,W}) = \int \left[\frac{\bar{G}_n(u) - \bar{G}_0(u)}{\bar{G}_n(u)} \right] [\bar{Q}_n^*(u) - \bar{Q}_0(u)] dQ_{0,W}(u) \ . 
\end{equation}
A key step in establishing asymptotic linearity of a TMLE is showing that $R_2$ is asymptotically negligible, i.e., $R_2(\bar{Q}_n, \bar{Q}_0, \bar{G}_n, \bar{G}_0, Q_{0,W}) = o_{\text{p}}(n^{-1/2})$. This requirement is satisfied if, for example, both $\bar{Q}_{n}^* - \bar{Q}_{0}$ and $\bar{G}_{n} - \bar{G}_{0}$ are $o_{\text{p}}(n^{-1/4})$ with respect to the $L^2(P_0)$ norm. If so, equation (\ref{linearization}) implies that the TMLE is asymptotically linear with influence function equal to the EIF and thus, by definition, the TMLE is efficient. Moreover, the central limit theorem implies that $n^{1/2}(\psi_n^* - \psi_0)$ converges in distribution to a mean-zero normal variate with variance $\sigma^2_0 := \E_{P_0}[ D(O \mid \bar{Q}_0, Q_{0,W}, \bar{G}_0)^2 ]$. 

We note an additional interesting feature of efficient estimators of the treatment-specific mean: they are doubly-robust. That is, the estimated treatment-specific mean is consistent for the true treatment-specific mean if either the estimated OR consistently estimates the true OR, the estimated PS consistently estimates the true PS, or both estimators are consistent. 

As with any asymptotic analysis, these results provide conditions under which estimators are well-behaved in large samples, but provide no guarantees of small-sample performance. In particular, in settings where the target estimand is weakly identifiable, in spite of the their asymptotic optimality, the TMLE and one-step may be unstable. For example, when the PS may assume very small values, the variance of the EIF may be large, which can cause erratic behavior of the EIF-based correction procedures. In the context of TMLE for the treatment-specific mean, this instability may manifest in the estimation the working model parameter $\epsilon$. The covariate $H_n$ in the parametric working model may have extremely large values, leading to a targeted OR estimator $\bar{Q}_n^*$ whose performance is considerably deteriorated relative to the initial OR estimator $\bar{Q}_n$. 

Often, the analyst has little prior information that would suggest whether or not such issues are present. Thus, we are motivated to consider automated procedures for constructing OR and PS estimators that are adaptive to near positivity violations. One such proposal is CTMLE. A CTMLE is based on a sequence of PS estimators that increase in complexity. For example, in the context of estimating the treatment-specific mean, we may start our sequence with an intercept-only logistic regression model and build a sequence ranging from that simple estimator to a flexible semiparametric estimator such as kernel regression. The sequence of candidate PS estimators of is used to generate a sequence of targeted ORs. The best of the targeted ORs is selected via cross-validation and is used to create a plug-in estimator. The principle underlying CTMLE is that the estimator searches for a reduced-dimension alternative to the true PS that is adaptive to how well the estimated OR fits the true OR. If the initial OR estimate provides a good fit, then there may be little benefit (or even detriment) to performing a TMLE correction based on a PS with extremely small values. On the other hand if the initial OR is a poor fit, then we may in fact benefit from such a correction. CTMLE can adapt to each of these situations. Because CTMLE is based on a sequence of increasingly nonparametric PS estimators, the procedure will, with probability tending to 1, select the last consistent estimator of the true PS in the sequence. Thus, in large samples, CTMLE is expected to have similar behavior to a standard TMLE. In this respect, CTMLE may be generally viewed as a procedure that offers finite-sample improvements over standard TMLE, while maintaining its asymptotic efficiency. 

The above discussion of efficiency and finite-sample considerations can be viewed in a more general lens than the context of nonparametric estimation of the ATE. In particular, these issues apply to a more broad set of problems that involve observed data structures that can be represented as a coarsened at random (CAR) version of a full data structure, while many models and parameters relevant to causal inference are special cases of this setting. The general CAR setting is discussed further in Appendix C. 


\section{Methods} \label{sec:methods}

We now propose a particular CTMLE for the treatment-specific mean that is robust to near positivity violations, but avoids the sequential PS estimator selection that is typical of other CTMLE proposals. The distinct aspect of the current proposal relative to previous CTMLE-based proposals is that we rely fully on $\bar{Q}_n$ converging to $\bar{Q}_0$ faster than $n^{-1/4}$ with respect to $L^2(P_0)$-norm. Because the OR estimator is consistent, any PS estimator will lead to a consistent estimate of the treatment-specific mean, due to the double-robustness of the EIF. However, our procedure asks for a more stringent property on the adaptive PS estimator. We require that the resulting CTMLE be asymptotically linear, and thereby maintain a Normal limiting distribution. The challenge in so-doing is that our selected PS estimator is generally inconsistent for the true PS. Previous work has shown that, even when a nonparametric estimate of the OR is consistent, inconsistent estimation of the PS can have serious implications for the behavior of one-step estimators and TMLEs \citep{vanderLaan14a,Benkeseretal17}. The issue stems from the fact that the second-order remainder is generally not asymptotically negligible. Thus, the key in achieving our goal is to choose an adaptive PS estimator such that the second-order remainder remains asymptotically negligible under reasonable conditions. In Theorem 1, we establish that this goal can be achieved by using an estimate of \[
	\bar{G}_0(w \mid \bar{Q}_0) := \mbox{pr}_{P_0}[A = 1 \mid \bar{Q}_0(W) = \bar{Q}_0(w)] 
\]
rather than an estimate of the true PS. In words, $\bar{G}_0(\cdot \mid \bar{Q}_0)$ describes the probability of receiving treatment as a function of the conditional mean outcome under treatment. This adaptive PS estimate is substituted into the usual TMLE (or one-step) procedures for estimation and inference. 

Our proposed CTMLE is implemented in the following steps: \begin{enumerate}
	\item {\it estimate OR}: regress $Y$ on $W$ amongst units observed to receive treatment $A = 1$ to obtain OR estimate $\bar{Q}_n$; 
	\item {\it predict outcome}: use estimated OR to obtain a prediction $\bar{Q}_n(W_i)$ for each observed data unit, $i = 1,\dots,n$;
	\item {\it estimate adaptive PS}: regress $A$ on predictions $\bar{Q}_n(W_i)$ to obtain adaptive PS estimate $\bar{G}_n(\cdot \mid \bar{Q}_n)$; 
	\item {\it predict PS}: use estimated PS to obtain prediction $\bar{G}_n(W_i \mid \bar{Q}_n)$ for each observed data unit, $i = 1, \dots, n$; 
	\item {\it fit OR working model}: fit logistic regression of outcome $Y$ on covariate $H_n(A, W) := A / \bar{G}_n(W \mid \bar{Q}_n)$ with offset $\mbox{logit}[\bar{Q}_n(W)]$; denote by $\epsilon_n^\#$ the estimated coefficient; 
	\item {\it target OR estimate}: use OR working model to obtain a prediction $\bar{Q}_n^\#(W_i) = \mbox{expit}\{\mbox{logit}[\bar{Q}_n(W_i)] + \epsilon_n^\# H(1, W_i)\}$ for each observed data unit, $i = 1,\dots,n$; 
	\item {\it compute plug-in estimate}: the CTMLE is $\psi_n^\# := n^{-1} \sum_{i=1}^n \bar{Q}_n^\#(W_i)$. 
\end{enumerate}
An analogous description of the collaborative one-step estimator is included in Appendix E. Sample \texttt{R} code to compute the estimators is included in Appendix F. 

The following theorem establishes the weak convergence of the proposed estimator. We explicitly discuss regularity conditions in Appendix G and strategies for weakening these conditions in Appendix F.
\begin{theorem} \label{thm1}
Under the regularity conditions in Appendix G, \[
	\psi_n^\# - \psi_0 = \frac{1}{n}\sum_{i=1}^n D^*(O_i \mid \bar{Q}_0, Q_{0,W}, \bar{G}_0(\cdot \mid \bar{Q}_0)) + o_{\text{p}}(n^{-1/2}) \ ,
\]
and $n^{1/2}(\psi_n^\# - \psi_0)$ converges in distribution to a mean-zero Normal variate with variance $\tau^2_0 := \E_{P_0}[D^*(O \mid \bar{Q}_0, Q_{0,W}, \bar{G}_0(\cdot \mid \bar{Q}_0))^2]$. 
\end{theorem}
The asymptotic variance $\tau^2_0$ of $\psi_n^\#$ is generally smaller than that of the standard TMLE $\sigma^2_0$, so that the proposed estimator is super efficient. That is, at any fixed data distribution in ${\cal M}$, this CTMLE will be asymptotically more efficient than the standard TMLE. 

\subsection{Variance estimation}
We propose to estimate the standard error of $\psi_n^\#$, based on 
a cross-validated estimate of the variance of the influence function. Specifically, consider a $V$-fold cross-validation scheme, wherein data are randomly partitioned into $V$ blocks of approximately equal size. For $v=1,\dots,V$, denote by $\mathcal{V}_v \subset \{1,\dots,n\}$ the indices of units in each block. For $i = 1,\dots,n$, we denote by $\bar{Q}_{n,v}^0(W_i)$ the predicted outcome for observation $i$ based on an OR estimate fit when observation $i$ was in the hold-out block. That is, to obtain $\bar{Q}_{n,v}^0$, we regress $Y_i$ on $W_i$ in units with $A_i = 1$ and $i \in \{1,\dots,n\} \setminus \mathcal{V}_v$. Then we use this fitted regression to obtain predictions based on $W_j, j \in \mathcal{V}_v$. Similarly, we denote by $\bar{G}_{n,v}^0(\cdot \mid \bar{Q}_{n,v}^0)$ the $v$-th estimated adaptive PS, $v=1,\dots,V$. This quantity is computed by regressing $A_i$ on $\bar{Q}_{n,v}^0(W_i)$ for $i \in \{1,\dots,n\} \setminus \mathcal{V}_v$. This regression is then used to obtain predictions based on $\bar{Q}_{n,v}^0(W_j), j \in \mathcal{V}_v$, which we denote by $\bar{G}_{n,v}^0(W_j \mid \bar{Q}_{n,v}^0)$. Finally, we denote by $Q_{n,W}^1$ the empirical distribution of $W$ based on $\{O_i : i \in \mathcal{V}_v\}, v = 1,\dots,V$. The cross-validated variance estimator is \begin{align*}
	\tau_{n,\text{cv}}^2 &:= \frac{1}{V} \sum_{v = 1}^V \bigg\{ \frac{1}{|\mathcal{V}_v|} \sum_{i \in \mathcal{V}_v} \biggl[ D^*(O_i \mid \bar{Q}_{n,v}^0, Q_{n,W}^1, \bar{G}_{n,v}(\cdot \mid \bar{Q}_{n,v}^0)) \\ 
	&\hspace{2in} - \frac{1}{|\mathcal{V}_v|}\sum_{j \in \mathcal{V}_v}  D^*(O_i \mid \bar{Q}_{n,v}^0, Q_{n,W}^1, \bar{G}_{n,v}(\cdot \mid \bar{Q}_{n,v}^0)) \biggr]^2 \bigg\} \ . 
\end{align*}

\noindent \textbf{Remark. } For notational simplicity, we have focused on presenting an estimate of the treatment-specific mean. An estimate of the ATE can thus be obtained by repeating the entire procedure but switching the labeling of the treatment. Alternatively, we also propose a CTMLE that directly targets the ATE in Appendix F. The major difference is that the adaptive PS now involves regression $A$ on both the OR in units who received treated, as well as the OR in units who did not receive treatment. 

\section{Simulations}

We evaluated the performance of the proposed collaborative estimators relative to their standard counterparts in two simulation studies. We focus our presentation on comparing CTMLE and TMLE results, while a comparison of the one-step estimators is included in Appendix G. The first simulation evaluated the relative performance of CTMLE vs. TMLE as a function of sample size and strength of positivity violations. In this setting the estimators of the OR and PS are based on correctly-specified parametric models. The results demonstrate the behavior the proposed estimators as a function both of sample size and strength of positivity violations. The second simulation offers a competitive setting for comparing the various estimators. In this setting, both the OR and PS are highly nonlinear functions of the covariates and involve complex covariate interactions. To consistently estimate these complex functions, we utilize the highly adaptive loss minimum loss estimator (HAL-MLE) \citep{vanderLaan15,Benkeser&vanderLaan16}. This estimator has been shown to the requisite regularity conditions of Theorem \ref{thm1} under extremely mild assumptions on the true nuisance parameters.  

In both simulation settings, we generated and analyzed 1000 Monte Carlo replicate data sets in each setting at each sample size. We evaluated the estimators on their bias, variance, and mean-squared error. We also present visualizations of the estimated sampling distributions of the scaled and centered sampling distributions. Further, we present the coverage probability of nominal 95\% Wald-style confidence intervals based on the Monte Carlo standard deviation of the estimators (i.e., an oracle confidence interval) and based on influence function-based standard error estimates. 

\subsection{Simulation 1}

\begin{figure}[ht]
\centering
\includegraphics[width = \textwidth]{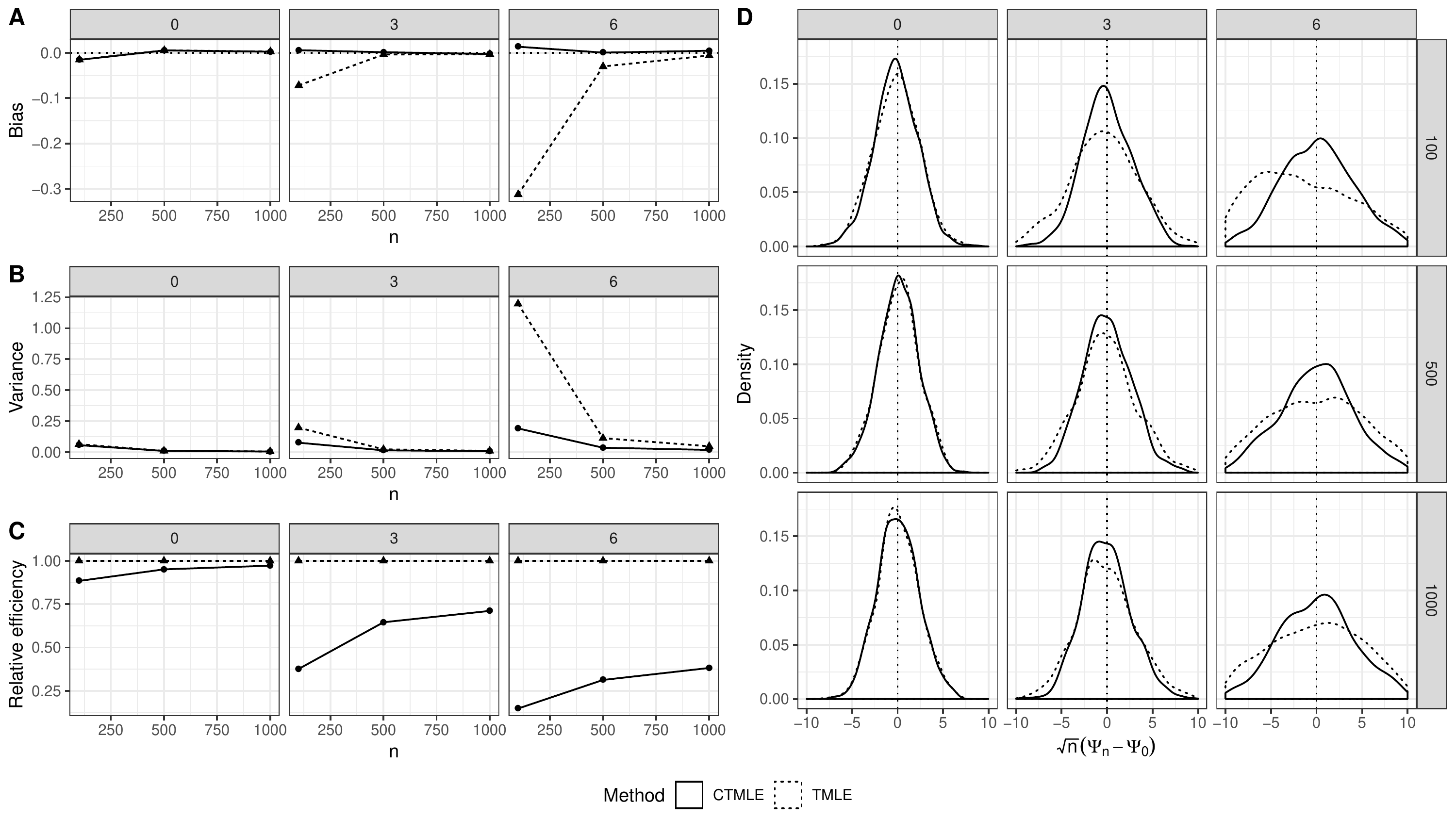}
\caption{Results for simulation 1 comparing CTMLE and TMLE. Each panel displays a different performance metric and each sub-panel displays results for $\gamma \in \{0,3, 6\}$, representing, respectively, settings with no positivity, moderate positivity, and extreme positivity violations. Panel A: Bias of the estimators. Panel B: Variance of the estimators. Panel C: Relative efficiency (defined as ratio of mean squared-error) of CTMLE to TMLE. Numbers below one indicate greater efficiency of CTMLE. Panel D: Kernel density estimates of sampling distributions using a Gaussian kernel and Silverman's rule of thumb bandwidth \citep{silverman1986density}.}
\label{tmle_estimator_results_sim1}
\end{figure}

For each sample size $n \in \{100, 500, 1000\}$ we generated data as follows. $W$ was an eight-variate vector. We drew the first seven components $(W_1,\dots,W_7)$ of $W$ from a from an Uniform distribution on $[-1.5, 1.5]^7$. The final component $W_8$ of $W$ was drawn from a Bernoulli(0.5) distribution. Given $W = w$, the treatment $A$ was drawn from a Bernoulli distribution with success probability \[
  \bar{G}_0(w) = \mbox{expit}\left(0.5 \gamma - \gamma w_8 + \sum_{j=1}^7 2^{1-j} w_j \right) \ . 
\]
Given $W=w$ and $A=a$, the outcome $Y$ was drawn from a Normal distribution with unit variance and mean \[
  \bar{Q}_0(a,w) = a - \sum_{j=1}^7 2^{1-j} w_j \ . 
\]
The true ATE in this setting is one. We induced positivity violations by choosing increasingly large values of $\gamma$. We evaluated three choices, $\gamma \in \{0, 3, 6\}$. These choices resulted in PS bounded in (0.05, 0.95), (0.01, 0.99), and (0.003, 0.997), respectively. The standard TMLE and one-step estimators used correctly-specified logistic regression for the PS and correctly-specified linear regression for the OR. The CTMLE and collaborative one-step used correctly-specified linear regression for the OR and HAL-MLE for the adaptive PS.

In settings with no positivity issues ($\gamma = 0$), we found that CTMLE and TMLE performed approximately equivalently, though CTMLE offered modest benefits at the smallest sample size (Figure \ref{tmle_estimator_results_sim1}). As $\gamma$ increased, propensity scores were pushed towards zero and one, and we saw increased performance of CTMLE relative to TMLE. CTMLE offered significant improvements both in terms of bias and variance, and was more than four times as efficient in the $\gamma = 6$, $n = 100$ scenario. The sampling distribution of both estimators was well approximated by the reference asymptotic distribution as indicated by nominal coverage of oracle confidence intervals (Figure \ref{tmle_coverage_results_sim1}, Panel A). However, while the estimated standard errors of the TMLE estimator performed well in larger samples, the estimated standard errors of CTMLE had poor performance, often underestimating the true variability of the estimator (Figure \ref{tmle_coverage_results_sim1}, Panel B). 

Results for the collaborative one-step vs. standard one-step estimator were essentially the same as for TMLE (Appendix XXX). 

\begin{figure}[ht]
\centering
\includegraphics[width = 0.5\textwidth]{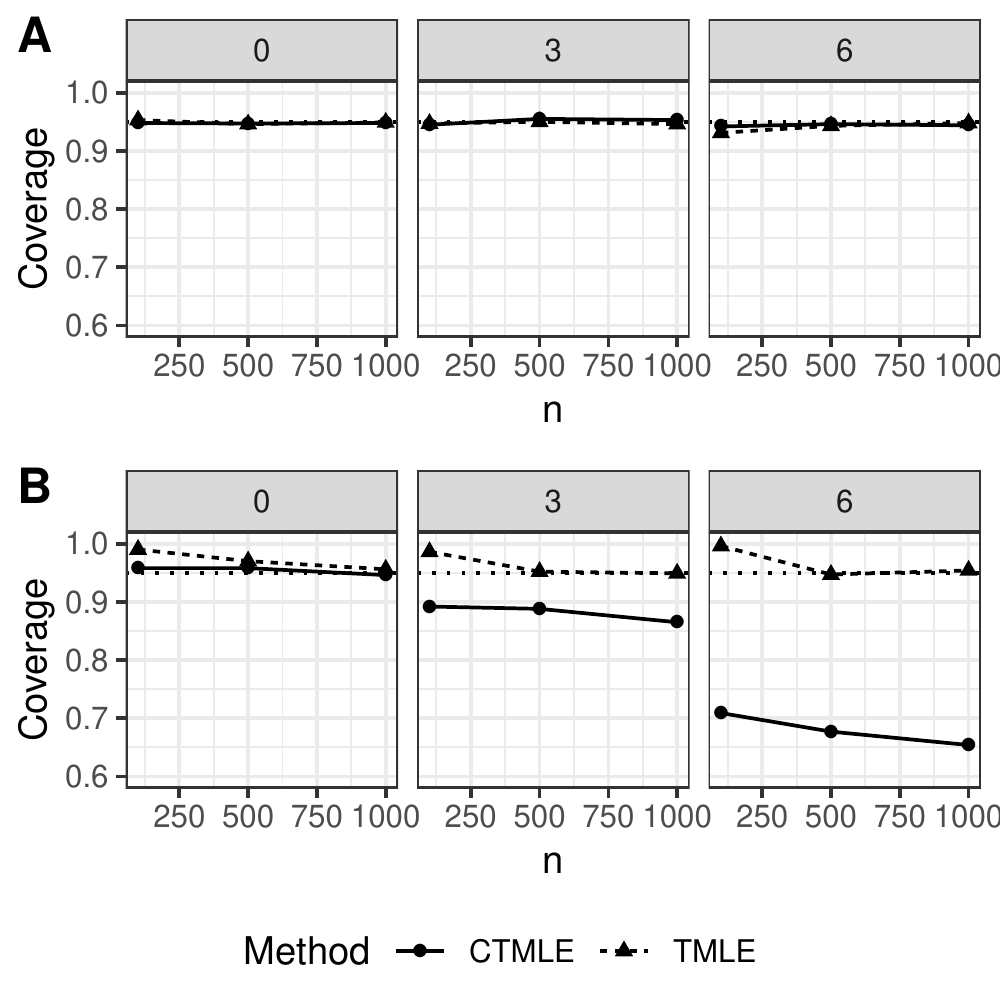}
\caption{Results for simulation 1 comparing confidence intervals for CTMLE and TMLE. Each panel displays the coverage as a function of sample size and each sub-panel displays results for $\gamma \in \{0,3, 6\}$. Panel A: Coverage probability of nominal 95\% oracle confidence intervals. Panel B: Coverage probability of nominal 95\% confidence intervals based on estimated standard errors.}
\label{tmle_coverage_results_sim1}
\end{figure}

\subsection{Simulation 2}

\begin{figure}[ht]
\centering
\includegraphics[width = 0.8\textwidth]{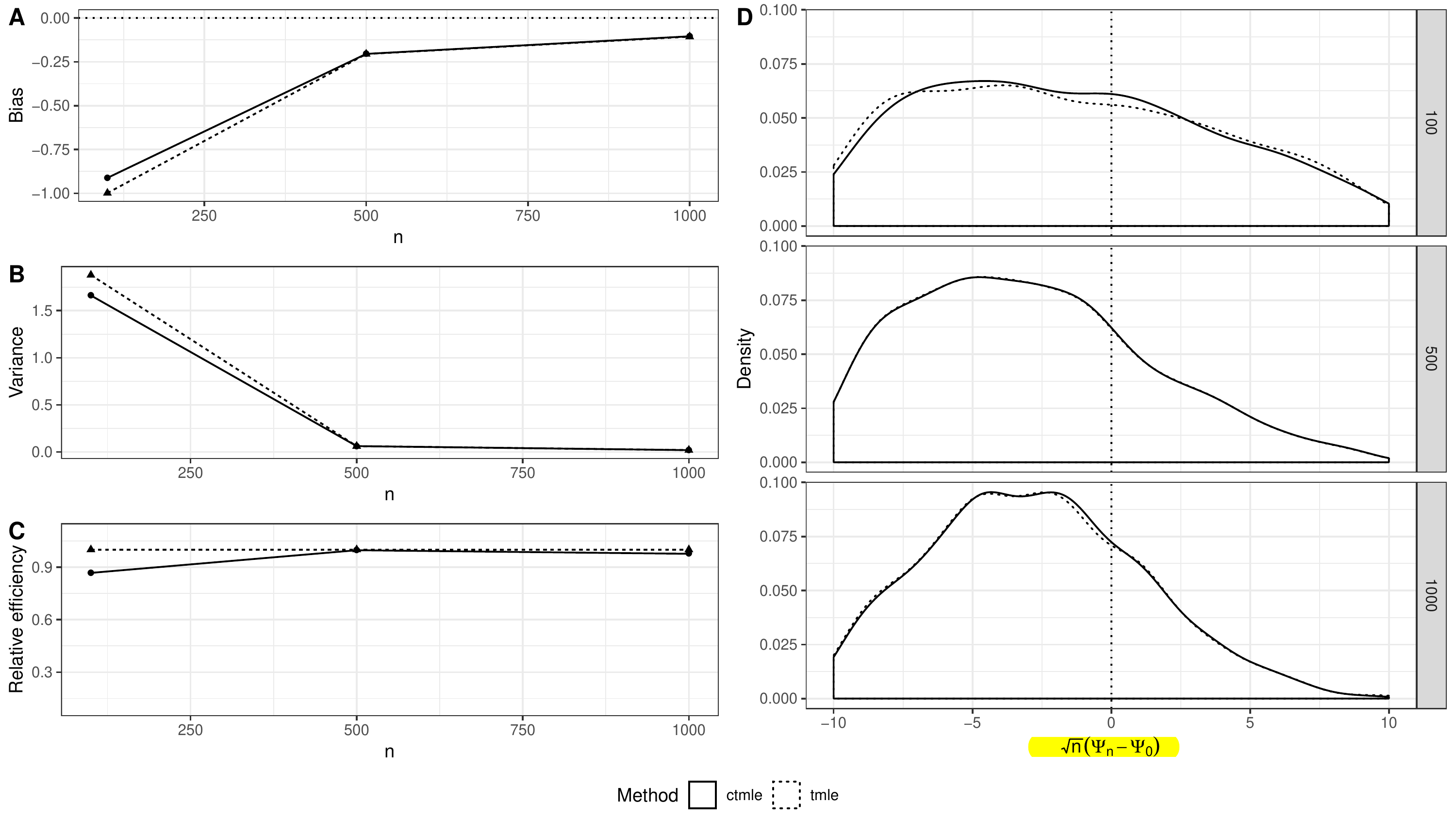}
\caption{Results for simulation 2 comparing CTMLE and TMLE. Each panel displays a different performance metric and each sub-panel displays results for $\gamma \in \{0,3, 6\}$, representing, respectively, settings with no positivity, moderate positivity, and extreme positivity violations. Panel A: Bias of the estimators. Panel B: Variance of the estimators. Panel C: Relative efficiency (defined as ratio of mean squared-error) of CTMLE to TMLE. Numbers below one indicate greater efficiency of CTMLE. Panel D: Kernel density estimates of sampling distributions using a Gaussian kernel and Silverman's rule of thumb bandwidth \citep{silverman1986density}.}
\label{tmle_estimator_results_sim2}
\end{figure}

We based this simulation setting on the oft-cited \citet{Kang:Schafer07} simulation design. This design is notoriously challenging for causal effect estimators due to extremely non-linear covariate relationships in the OR and PS and highly complex interactions between covariates. In our simulation, we drew $Z_1$ from a Uniform(0.5, 2) distribution and drew $Z_2,\dots,Z_5$ from a Uniform distribution on $[-2,2]^4$. Given $Z = z$, the treatment $A$ was drawn from a Bernoulli distribution with success probability \[
  \bar{G}_0(z) = \mbox{expit}(- Z_1 + 0.5  Z_2 - Z_3 - 0.1  Z_4 + Z_5 + 0.75Z_5^2) \ . 
\]
Given $Z = z$ and $A = a$, the outcome $Y$ was drawn from a Normal distribution with unit variance and mean \[
  \bar{Q}_0(a,z) = 210 + 27.4  Z_1 + 13.7  Z_2 + 13.7  Z_3 + 13.7  Z_4 \ . 
\]
The true ATE in this setting is zero and the true PS is bounded between (0.004, 0.999). A challenge of this simulation setting is that the covariates $Z$ are not available to the analyst. Instead, we must base our estimation on $W = W_1, \dots, W_5$, which is generated from $Z$ as follows \begin{align*}
  W_1 &= \mbox{exp}(Z_1/2) & W_2 &= Z_2/[1 + \mbox{exp}(Z_1)] + 10 \\
  W_3 &= [(Z_1 Z_3)/25 + 0.6]^3 & W_4 &= (Z_2 + Z_4 + 20)^2 \\
  W_5 &= Z_5 \ . 
\end{align*}
As such, the true PS and OR, when expressed as functions of $W$ are highly non-linear and involve interactions between the various components of $Z$. To estimate these functions well, we require extremely flexible regression tools. Thus, we use HAL-MLE for the PS and OR used by the TMLE and one-step estimators. Similarly, we used HAL-MLE for the OR and adaptive PS. 

\begin{figure}[ht]
\centering
\includegraphics[width = 0.4\textwidth]{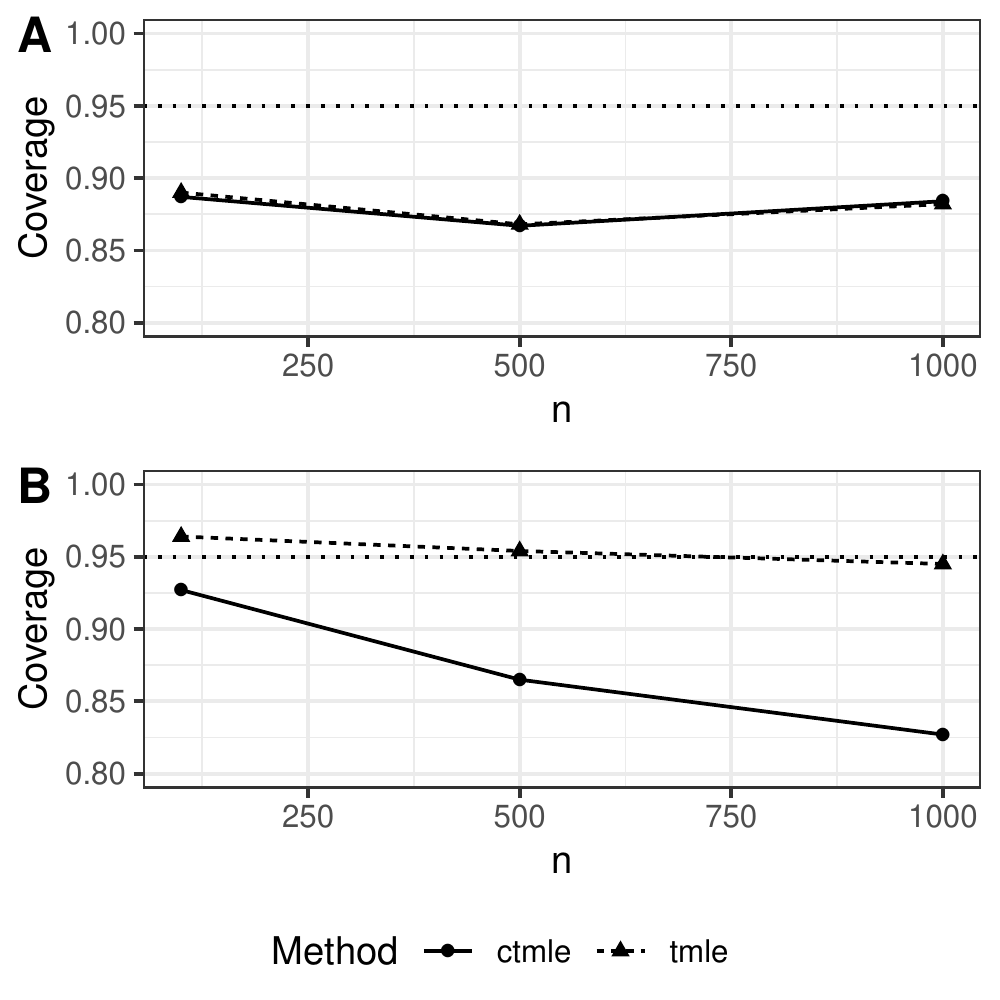}
\caption{Results for simulation 1 comparing confidence intervals for CTMLE and TMLE. Each panel displays the coverage as a function of sample size and each sub-panel displays results for $\gamma \in \{0,3, 6\}$. Panel A: Coverage probability of nominal 95\% oracle confidence intervals. Panel B: Coverage probability of nominal 95\% confidence intervals based on estimated standard errors.}
\label{tmle_coverage_results_sim2}
\end{figure}

As expected, both the the TMLE and CTMLE struggled in this very challenging simulation study (Figure \ref{tmle_estimator_results_sim2}). While CTMLE offered modest benefits in terms of variance, the bias of the two estimators was comparable. Nevertheless, we do see evidence of asymptotic linearity of both estimators in that the sampling distributions of both of the scaled and centered estimators appear to be moving towards an appropriate center at zero. Nevertheless, both estimators had relatively large bias, even in the largest sample size $n = 1000$, as shown by the less-than-nominal coverage of the oracle confidence intervals (Figure \ref{tmle_coverage_results_sim2}, Panel A). The cross-validated influence function-based variance estimators overestimated the variability of the TMLE, which resulted in near nominal coverage for those intervals (Panel B). However, as in simulation 1, we found that the proposed variance estimators for the CTMLE significantly underestimated the variability of the estimator, resulting in poor coverage. 

\section{Data Analysis}

We analyzed data collected via the Cebuano Longitudinal Health and Nutrition Survey (CLHNS) \citep{adair2010cohort}. CLHNS is an ongoing study of a cohort of Filipino women who gave birth between May 1, 1983, and April 30, 1984. Children born to these women in that period have been followed through subsequent surveys over multiple years. We used these data to estimate the effect of term pregnancy (pre-, full-, post-) on children's schooling achievement in Cebuano, English, and mathematics. The putative confounders we considered were related to parental characteristics (maternal/paternal age, maternal height, maternal/paternal education, maternal age at first birth, maternal parity, maternal marital status), household characteristics (number of children in the household, child/adult ratio, child dependency ratio, crowding index, number adult males/females in household, urbanicity score, water availability, sanitation), and socioeconomic information (socioeconomic status, familial health care access, total family income). Our strategy for handling missing covariate data is described in Appendix H. Pre-, full-, and post-term pregnancies were defined as those less than 37 completed weeks, between 37 and 41 completed weeks, and greater than 41 completed weeks, respectively. Our outcome is the standardized total performance on the National Elementary Assessment Test (NEAT) that was collected in 1994-95. Children were scored on their achievement in Cebuano, English, and mathematics. We standardized and summed these three scores to create a composite school achievement outcome. We analyzed a total of $n = 2150$ children who completed these tests. 

A-priori, we may not expect positivity issues with these data since previous studies have not uncovered strong predictors of birth term \citep{di2011maternal}. On the other hand, this same research has highlighted significant geographic variability in predictors of birth term. Thus, when proposing a pre-specified analysis in a population with little prior information available on predictors of term birth, we may worry about the potential for near positivity violations. Consequently, we may wish to select a method that hedges against such violations, like the proposed CTMLE. 

We analyzed these data using TMLE, CTMLE, one-step, and collaborative one-step. For the outcome regression, we used a super learner. Super learner, a generalized version of regression stacking \citep{Wolpert92,Breiman96d}, is a cross-validation-based ensemble approach that creates a regression fit based on a weighted combination of candidate regression fits. The method is implemented in the \texttt{SuperLearner} \texttt{R} package, freely available through the Comprehensive R Archive Network (CRAN) \citep{SuperLearnerPackage}. Our candidate regressions for the OR and PS included generalized linear models (function \texttt{SL.glm} in the \texttt{SuperLearner} package), polynomial multivariate regression splines (\texttt{SL.earth}), random forests (\texttt{SL.ranger}), lasso (\texttt{SL.glmnet}), gradient boosted regression trees (\texttt{SL.gbm}), an intercept-only regression (\texttt{SL.mean}), and a forward stepwise generalized linear model (\texttt{SL.step.forward}). Each super learner was based on five-fold cross-validation. We used the HAL-MLE to estimate each adaptive PS utilized by the proposed CTMLE and collaborative one-step. There were three such scores: one each for the OR in pre-, full-, and post-term births. We constructed Wald-style confidence intervals using influence function-based standard errors as described above. We also tested the null hypothesis that the average schooling achievement was equivalent across the three categories of birth term using a two-degree-of-freedom Wald-style test. Finally, for comparison, we estimated each of the relevant quantities using a main-terms linear regression model, with nonparametric bootstrap confidence intervals (based on 500 bootstrap samples) and computed a p-value using a two-degree-of-freedom Wald-style test based on a sandwich variance estimator \citep{white1980heteroskedasticity}. 

\begin{table}[ht]
\centering
\begin{tabular}{lcccc}
  \hline
Method & Pre-term & Full-term & Post-term & p-value \\ 
  \hline
  TMLE & -0.03 (-0.14, 0.07) & 0.02 (-0.03, 0.06) & -0.08 (-0.21, 0.06) & 0.33 \\ 
  CTMLE & -0.03 (-0.13, 0.07) & 0.02 (-0.03, 0.06) & -0.07 (-0.20, 0.06) & 0.37 \\ 
  OS & -0.03 (-0.14, 0.08) & 0.02 (-0.03, 0.06) & -0.07 (-0.21, 0.06) & 0.36 \\ 
  COS & -0.03 (-0.13, 0.07) & 0.02 (-0.03, 0.06) & -0.07 (-0.19, 0.06) & 0.39 \\ 
  LM & -0.03 (-0.13, 0.07) & 0.02 (-0.03, 0.06) & -0.18 (-0.46, 0.08) & 0.27 \\ 
   \hline
\end{tabular}
\caption{Estimated average standardized schooling achievement score (95\% confidence interval) for pre-, full-, and post-term births using various methods. TMLE = targeted minimum loss estimator, CTMLE = the proposed collaborative TMLE, OS = one-step estimator, COS = collaborative OS estimator, LM = estimator based on a main terms linear model. The p-value is from a test of the null hypothesis that the three means are equal.}
\label{cebu_results}
\end{table}

Each of the methods provided similar point estimates, except the estimate of the average standardized score in post-term births was lower when considering the linear model-based estimator (Table \ref{cebu_results}). In each case, we fail to reject the null hypothesis of equal school achievement by birth term category at any reasonable type-one error threshold. That the standard efficient estimators gave similar results to their collaborative counterparts is not surprising considering that no significant positivity issues were uncovered in the analysis. In particular, the estimated probabilities of pre-, full-, and post-term birth for the observed participants were bounded between (0.12, 0.23), (0.68, 0.81), and (0.06, 0.10), respectively. On the other hand, the corresponding estimated adaptive PS's were bounded between (0.14, 0.17), (0.78, 0.78), and (0.07, 0.08), respectively. 

\section{Discussion}

It has been recognized in the literature that efficient estimators such as TMLE and one-step can show erratic, non-robust behavior if the target estimand is weakly identifiable. This setting is most often seen in observational studies with high-dimensional covariates with little a-priori knowledge of which covariates are confounders of the treatment/outcome relationship. In these settings various CTMLE estimators have been proposed. Relative to these existing estimators, our proposed estimator enjoys some important benefits. First, it avoids many of the decision points needed in a typical CTMLE implementation. The user need only choose a regression estimator for the OR and for the adaptive PS. Beyond those choices, standard TMLE software can be used simply by substituting an estimate of the adaptive PS for the true PS. Our theorem establishes similar weak convergence results to those that have been previously proven for TMLE. Another potential benefit of our proposal is that, while the statistical motivation for our choice of adaptive PS is somewhat technical -- ensuring a second-order remainder of the linearized estimator is asymptotically negligible -- the interpretation of the adaptive PS is easily explained to applied practitioners. For example, in a medical context, rather than estimating the true PS, which describes a patient's propensity for receiving treatment as a function their medical history and other confounders, we instead opt for the adaptive PS, which describes the propensity for receiving treatment as a function treated patients' estimated risk of disease. 

We expect that the proposed CTMLE will exhibit more robust finite-sample performance relative to standard TMLE in settings where the target estimand is weakly identifiable, and indeed the results of our first simulation explicitly demonstrated this phenomenon. On the other hand, because the proposed CTMLE is an \emph{irregular estimator}, it may perform poorly for certain data generating distributions. A more careful comparison between standard TMLE and our super-efficient CTMLE is thus warranted. First, it is relevant to note that under lack of positivity it might be harder to estimate the OR well and both the small- and large-sample performance of our estimator is very much tied to the performance of the OR estimator. While recent developments such as HAL-MLE theoretically ensure that the $n^{-1/4}$-consistency requirement is satisfied under weak conditions, there are no finite-sample guarantees of adequate performance. Therefore, it would seem particularly beneficial for our proposed estimator that a flexible estimation framework be employed for the OR. For example, a cross-validation-based regression stacking approach, such as the super learner approach that we used in the data analysis may be particularly beneficial. Nevertheless, contrary to other CTMLE approaches that enjoy collaborative double-robustness, our CTMLE is not doubly-robust: we cannot compensate for an inconsistent OR estimator by using a consistent PS estimator. However, our proposed estimator appears to completely avoid any risk of having a TMLE in which the targeting step decrements the behavior of the initial OR estimator. As a final point of comparison, it appears that variance estimation for this CTMLE is a more challenging task than for standard TMLE, as evidenced by the poor behavior of the confidence intervals in both simulations. It will thus be important in future work to consider alternative strategies for estimating variance (e.g., based on a bootstrap schema). 

Another strategy for providing additional robustness against poor estimation in small-samples is to utilize TMLE frameworks designed for doubly-robust inference \citep{vanderLaan14a,Benkeseretal17}. These estimators involve fitting additional parametric working models for both the OR and PS estimators beyond a typical TMLE implementation. The resultant effect estimator is asymptotically linear even if one of the OR or PS is inconsistently estimated (or is consistently estimated at too slow a rate). The additional fitting aims to minimize the contribution of the second-order remainder to the behavior of the estimator -- essentially the same goal as with our CTMLE proposal. Even when both the OR and adaptive PS estimators are consistent at rates faster than $n^{-1/4}$, this approach could result in finite-sample improvements for a TMLE or efficient CTMLE. We leave to future study whether and how such strategies can be used to further improve finite-sample behavior of estimators. 

Overall, we conclude that the relative finite-sample performance of an efficient TMLE/CTMLE (possibly modified to accommodate doubly-robust inference) and our proposed super-efficient CTMLE will depend on the particular type of data distribution. Clearly, if one has a-priori knowledge of the propensity score, in particular that it does not suffer from positivity issues, then an efficient TMLE or CTMLE strategy may be the favorable estimator. On the other hand, if the PS is poorly understood and/or the parameter is weakly identifiable, then the super-efficient CTMLE may be a better option. 

In future work, we will generalize our asymptotic linearity results to the more general CAR setting described in Appendix B. Such extensions would allow us to tackle other challenging problems in causal inference, such as estimation of the counterfactual mean of a treatment administered at several timepoints subject to time-varying confounding. 

\bibliography{refs}

\appendix 
\section*{Appendix A. Regular estimators}
We can view an estimator $\psi_{n}$ of $\psi$ as the byproduct of an algorithm $\hat{\Psi}$ applied to independent
observations $O_{1}, \ldots, O_{n}$ drawn from a distribution $P \in \mathcal{M}$. We consider a fluctuation $\{P_{h} :  h \in  H\}$, where $H$ is an index set, satisfying \[
\frac{d}{dh} \mbox{log}\left[\frac{dP_h}{dP}(O)\right]\bigg|_{h=0} = s(O) \ ,
\]  
with $\E_{P} [s(O)] = 0$ and $\E_{P}  [s^2(O)] <  \infty$. We say that $\psi_n$ is 
a regular estimator of $\psi$ at $P$ if, for any such fluctuation, the estimator $\psi_{n,n^{-1/2}}$ of the parameter $\Psi(P_{n^{-1/2}})$ obtained  by applying the algorithm  $\hat{\Psi}$ on independent
observations $O_{1}$, \ldots, $O_{n}$ drawn from $P_{n^{-1/2}}$ is such that
$n^{1/2} [\psi_{n,n^{-1/2}} - \Psi(P_{n^{-1/2}})]$  converges in distribution to a random variate 
whose distribution does not depend on $s$. See Chapter 24 of \citet{Vaart98} for more on regularity in the context of efficient estimation. 

\section*{Appendix B. Linearization of plug-in estimator}
Below, we make use of empirical process notation, writing $Pf$ to denote $\int f(o) dP(o)$ for a given $P$-integrable function $f$ and for each $P\in\mathcal{M}$. We also denote by $P_n$ the empirical distribution function based on $O_1,\ldots,O_n$, so $P_n f = n^{-1} \sum_{i=1}^n f(O_i)$. A linearization of the parameter allows us to write \begin{align*}
\Psi(Q_n^*)-\Psi(Q_0)\ &=\ -P_0 \Dstarn + \Remstar \\
&= (P_n - P_0) \Dstarn - P_n \Dstarn \\
&\hspace{2.4in} + \Remstar \\ 
&= (P_n - P_0) \Dstarlim - P_n \Dstarn \\
&\hspace{0.2in} + (P_n - P_0) [\Dstarn - \Dstarlim] \\
&\hspace{2.4in} + \Remstar  \ . 
\end{align*}
where $\Remstar$ is defined in (\ref{remainder}). We assume \begin{enumerate}
	\item[(i)] $P_0 \Dstarlim = 0$; 
	\item[(ii)] $P_n \Dstarn = o_{\text{p}}(n^{-1/2})$; 
	\item[(iii)] $(P_n - P_0) [\Dstarn - \Dstarlim] = o_{\text{p}}(n^{-1/2})$. 
\end{enumerate}
Assumption (i) stipulates that the EIF have mean zero, which is satisfied, for example, if either $\bar{Q}^* = \bar{Q}_0$ or $\bar{G} = \bar{G}_0$. Assumption (ii) stipulates that the EIF estimating equation is approximately solved. Because $\bar{Q}_n^*$ is a targeted estimate of $\bar{Q}_0$, this assumption is trivially satisfied, since $P_n \Dstarn = 0$. The third assumption is satisfied if, for example, $P_0[ \Dstarn - \Dstarlim ]^2$ converges in probability to zero and $\Dstarn - \Dstarlim$ falls in $P_0$-Donsker class with probability tending to 1. 

\section*{Appendix C. Generalization to coarsened at random (CAR) data structures}
In this Appendix, we describe how the ideas presented in the main text generalize to coarsened at random (CAR) data structures. 
Suppose we observe $n$ i.i.d. copies $O_1,\ldots,O_n$ of $O=\Phi(C,X)\sim P_0\in {\cal M}$ for a full-data random variable $X$, censoring variable $C$, and a many to one-mapping $\Phi$. We assume that the conditional distribution $G$ of $C$, given $X$, satisfies the CAR assumption. Our model ${\cal M}=\{P_{P_X,G}:P_X\in {\cal M}^F,G\in {\cal G}\}$ for the observed data is implied by the full-data  model ${\cal M}^F$ for the full-data distribution $P_{X,0}$ and censoring model ${\cal G}$ for the true conditional distribution $G_0$ of censoring. Thus, the data distribution $P_0$ is parameterized as $P_0=P_{P_{X,0},G_0}$. A full-data target parameter is defined as a mapping $\Psi^F:{\cal M}^F\rightarrow\openr^d$ so that $\psi^F_0=\Psi^F(P_{X,0})$ is the target estimand of interest. We assume that $\Psi^F$ is pathwise differentiable at $P_X\in {\cal M}^F$ with EIF relative to this model $D^F(\cdot \mid P_X)$. We define the exact second-order remainder as 
\[
R_2^F(P_X,P_{X,0}) = \Psi^F(P_X)-\Psi^F(P_{X,0})+P_{X,0} D^F(\cdot \mid P_X) \ .
\]
We assume that the full-data parameter is identifiable so that there exists a $\Psi:{\cal M} \rightarrow \bm{\Psi}$, where $\bm{\Psi}$ is the parameter space, so that $\Psi(P_{P_X,G}) = \Psi^F(P_X)$ for all $P_{P_X,G}\in {\cal M}$.
The statistical estimand is thereby defined as $\Psi(P_0)$. We assume that $\Psi$ is a pathwise differentiable function of $P$ with EIF under $P$ relative to $\mathcal{M}$, $D^*(P)=D^*(P_X,G)$. Such parameters allow an expansion $\Psi(P)-\Psi(P_0)=(P-P_0)D^*(\cdot \mid P)+R_2(P,P_0)$, where \[
R_2(P,P_0) = \Psi(P) - \Psi(P_0) + P_0D^*(\cdot \mid P)
\] 
is the exact second order remainder for the observed data target parameter.

Suppose that the target estimand $\Psi(P)$ depends on $P$ through a functional $P \rightarrow Q(P)$, and that we have a loss function $L(\cdot \mid Q)$ such that that $Q_0 = \arg\min_{Q \in \mathcal{Q}} P_0 L(\cdot \mid Q)$, and 
$\mathcal{Q} = \{Q(P):P\in {\cal M}\}$ is the parameter space for $Q$. Therefore, with an abuse of notation, we may denote the statistical target parameter as $\Psi(Q)$, i.e., a function of $Q(P)$. Though $Q$ is a parameter of the observed data distribution, it is is also a parameter of the full-data distribution $P_X$, so it can be viewed both as a mapping $P\rightarrow Q(P)$ as well as $P_X \rightarrow Q(P_X)$. 

In addition, let this parameter $P\rightarrow Q(P)$ be chosen so that the EIF $D^*(\cdot \mid P)$ of the target parameter $\Psi$ under $P$ relative to the model $\mathcal{M}$ depends on the observed data distribution $P$ through $Q(P)$ and the censoring mechanism $G$, so that we can parameterize the EIF as $D^*(\cdot \mid Q,G)$. Similarly, we may denote the second-order remainder by $R_2(Q,Q_0,G,G_0)$ to emphasize that it involves a difference between $(Q,G)$ and $(Q_0,G_0)$.

Given an initial estimator $G_n$ of $G_0$ and $Q_n$ of $Q_0$, a TMLE of $\Psi(Q_0))$ defines a least-favorable parametric submodel $\{Q_{n,\epsilon}:\epsilon\}\subset \mathcal{Q}$ through $Q_n$ so that the loss-based score $\frac{d}{d\epsilon}L(\cdot \mid Q_{n,\epsilon})$ at $\epsilon =0$ along this path spans the estimated EIF $D^*(\cdot \mid Q_{n},G_n)$. Such a path is called a local least favorable path, while a universal least favorable path requires that this score at any $\epsilon$ equals $D^*(\cdot \mid Q_{n,\epsilon},G_n)$ \citep{vanderLaan&Gruber15}. 
The TMLE estimates the fluctuation parameter $\epsilon$ of this least favorable path with standard minimum loss-estimation $\epsilon_n=\arg\min_{\epsilon}P_n L(\cdot \mid Q_{n,\epsilon})$, and the corresponding one-step TMLE of the target estimand is the corresponding substitution estimator $\Psi(Q_n^*)$, where $Q_n^*=Q_{n,\epsilon_n}$. 
If the TMLE uses a universal least favorable submodel, then $P_n D^*(\cdot \mid Q_{n}^*,G_n) = 0$ \citep{vanderLaan&Gruber15}, but, even when one uses a local least favorable submodel, if $Q_n$ and $G_n$ are $n^{-1/4}$-consistent estimators, then, under regularity conditions, $P_n D^*(Q_n^*,G_n)=o_{\text{p}}(n^{-1/2})$ (see theorem in Appendix of \citep{vanderLaan15}). 
Thus, by similar arguments to those in Appendix A (i.e., appealing to $L^2(P_0)$-consistency of the EIF), \[
\Psi(Q_{n}^*)-\Psi(Q_0) = P_n D^*(\cdot \mid Q^*,G) + R_2(Q_n^*,Q_0,G_n,G_0) + o_{\text{p}}(n^{-1/2})\ , 
\]
where $Q^*$ is the in-probability limit of $Q_n^*$ and $G$ is the in-probability limit of $G_n$. As in the case of the treatment-specific mean, we can generally argue that $R_2$ is asymptotically negligible whenever both $Q_n^* - Q_0$ and $G_n - G_0$ are $o_{\text{p}}(n^{-1/4})$ in $L^2(P_0)$ norm. Under these conditions, the TMLE $\Psi(Q_n^*)$ is asymptotically efficient. 

Collaborative TMLE is a general tool whereby an estimator of the censoring mechanism $G_0$ is constructed sequentially, Starting with an initial estimator, we iteratively select an update among a set of candidates by maximizing the gain in empirical risk during the corresponding TMLE step. Having built a sequence of such candidate censoring mechanism estimators (and corresponding TMLE), we evaluate the cross-validated risk of each candidate to select the best TMLE among the candidate TMLEs of $Q_0$. The principle of CTMLE is that by maximizing the gain in empirical risk during the TMLE step, we are targeting a dimension-reduced quantity $G_0(\cdot \mid Q_n)$ that adjusts for a rich enough set of variables so that the TMLE stays consistent even when $Q_n$ is inconsistent. There is a whole class of such targets $G_0(Q_n)$ which will result in this consistency. This preservation of consistency of the C-TMLE for a $(Q_n,G_n(Q_n))$ that converges to a $(Q,G_0(Q))$ with possibly $Q\not =Q_0$ and $G_0(Q)\not =G_0$, is referred to as collaborative double robustness \citep{vanderLaan:Gruber10,Gruber&vanderLaan10a}. 

Our proposed CTMLE for the treatment-specific mean avoids iterative building of the fit of the censoring mechanism based on gain during the TMLE-step, but fully relies on $Q_n$ being $n^{-1/4}$-consistent. This strategy may be generalized to the broader CAR setting by choosing an adaptive target $G_0(\cdot \mid Q_n)$ so that the second order remainder \[
R_2(Q_n^*,Q_0,G_n(\cdot \mid Q_n),G_0)=o_{\text{p}}(n^{-1/2}). 
\]
Such a CTMLE is asymptotically linear with influence function equal to $D^*(\cdot \mid Q_0, G_0(\cdot \mid Q_0))$, which will generally have smaller variance than the EIF $D^*(\cdot \mid Q_0, G_0)$. We leave to future work the identification of general strategies for selecting adaptive censoring mechanism target parameters.

\section*{Appendix D. Additional estimators}

\subsection*{Collaborative one-step estimator}
The collaborative one-step estimator can be implemented in the following steps: \begin{enumerate}
	\item {\it estimate OR}: regress $Y$ on $W$ amongst units observed to receive treatment $A = 1$ to obtain OR estimate $\bar{Q}_n$; 
	\item {\it predict outcome}: use estimated OR to obtain a prediction $\bar{Q}_n(W_i)$ for each observed data unit, $i = 1,\dots,n$;
	\item {\it estimate adaptive PS}: regress $A$ on predictions $\bar{Q}_n(W_i)$ to obtain adaptive PS estimate $\bar{G}_n(\cdot \mid \bar{Q}_n)$; 
	\item {\it predict PS}: use estimated PS to obtain prediction $\bar{G}_n(W_i \mid \bar{Q}_n)$ for each observed data unit, $i = 1, \dots, n$; 
	\item {\it evaluate influence function}: use estimated OR and PS to compute $D^*(O_i \mid \bar{Q}_n, Q_{n,W}, \bar{G}_n(\cdot \mid \bar{Q}_n))$ for $i = 1,\dots,n$; 
	\item {\it compute one-step estimate}: $\psi_{n,\text{os}}^\# := n^{-1} \sum_{i=1}^n \bar{Q}_n(W_i) + n^{-1} \sum_{i=1}^n D^*(O_i \mid \bar{Q}_n, \bar{Q}_{n,W}, \bar{G}_n(\cdot \mid \bar{Q}_n))$. 
\end{enumerate}

\subsection*{Collaborative TMLE that targets the ATE directly}

Here we propose a TMLE that can be used to directly estimate the ATE. For each $w \in \mathcal{W}$, we denote by $\bar{Q}^1_n(w)$ an estimate of the OR in the treated $\E_{P_0}(Y \mid A = 1 , W = w)$ and similarly denote by $\bar{Q}_n^1(w)$ an estimate of the OR in the untreated $\E_{P_0}(Y \mid A = 0 , W = w)$. 
\begin{enumerate}
  \item {\it estimate OR}: regress $Y$ on $(A,W)$ to obtain regression estimate $(a,w) \mapsto \bar{Q}_n(a,w)$ (allowing an overload of the $\bar{Q}_n$ notation); 
  \item {\it predict outcome}: use estimated OR to obtain a prediction for each data unit setting $A = 1$ and $A = 0$; $\bar{Q}^1_n(W_i) = \bar{Q}_n(1,W_i)$ and $\bar{Q}_n^0(W_i) = \bar{Q}_n(0, W_i)$, $i = 1,\dots,n$;
  \item {\it estimate adaptive PS}: regress $A$ on predictions $(\bar{Q}^1_n(W_i), \bar{Q}^0_n(W_i))$ to obtain adaptive PS estimate $w \mapsto \bar{G}_n(w \mid \bar{Q}^1_n, \bar{Q}_n^0)$; 
  \item {\it predict PS}: use estimated PS to obtain prediction $\bar{G}_n(W_i \mid \bar{Q}_n^1, \bar{Q}_n^0)$ for each observed data unit, $i = 1, \dots, n$; define $\ell\bar{G}_{n}(A_i, W_i) = A_i \bar{G}_n(W_i \mid \bar{Q}_n^1, \bar{Q}_n^0) + (1 - A_i) 1 - \bar{G}_n(W_i \mid \bar{Q}_n^1, \bar{Q}_n^0)$; 
  \item {\it fit OR working model}: fit logistic regression of outcome $Y$ on covariate $H_n(A, W) := (2A - 1) / \ell\bar{G}_{n}(A, W)$ with offset $\mbox{logit}[\bar{Q}_n(A,W)]$; denote by $\epsilon_n^{\circ}$ the estimated coefficient; 
  \item {\it target OR estimate}: use OR working model to obtain predictions $\bar{Q}_n^{\circ}(1, W_i) = \mbox{expit}\{\mbox{logit}[\bar{Q}_n(1, W_i)] + \epsilon_n^{\circ} \ell\bar{G}_n(1, W_i) \}$ and $\bar{Q}_n^{\circ}(0, W_i) = \mbox{expit}\{\mbox{logit}[\bar{Q}_n(0, W_i)] + \epsilon_n^{\circ} \ell\bar{G}_n(0, W_i) \}$ for each observed data unit, $i = 1,\dots,n$; 
  \item {\it compute plug-in estimate}: the CTMLE of the ATE $\psi_n^circ := n^{-1} \sum_{i=1}^n [\bar{Q}_n^{\circ}(1, W_i) - \bar{Q}_n^{\circ}(0, W_i)]$. 
\end{enumerate}

\subsection*{Cross-validated CTMLE}

By using sample splitting within the TMLE framework, the so-called cross-validated TMLE (CV-TMLE), we can avoid the Donsker class condition in regularity condition (iv) of Appendix E below. Here, we outline how one can construct a super-efficient CV-CTMLE. The proof of the asymptotic linearity and thereby weak convergence of the estimator is completely analogous to the general proof of asymptotic efficiency of the CV-TMLE \citep{Zheng&vanderLaan12} combined with the specific proof presented in Appendix E. 

Consider $V$-fold sample splitting so that we have $V$ splits of the sample in a training and complementary validation sample. We denote the empirical distributions of the training and validation sample by $P_{n,v}^0$ and $P_{n,v}^1$, respectively, $v=1,\ldots,V$. As in Section \ref{sec:methods}, we denote by $\bar{Q}_{n,v}^0$, $\bar{G}_{n,v}^0(\cdot \mid Q_{n,v}^0)$ the estimators of the relevant nuisance quantities obtained in the $v$-th training sample. 
For a given split $v$, we define a working model for the OR, $\{\bar{Q}_{n,v,\epsilon}^1 = \mbox{expit}[\mbox{logit}(\bar{Q}_{n,v}^0) + \epsilon H_{n,v}] : \epsilon \in \mathbb{R}\}$, where $H_{n,v}^1(A_i, W_i) = A_i / \bar{G}_n^0(W_i \mid \bar{Q}_{n,v}^0)$. The MLE of $\epsilon$ is found by pooling over the validation samples to minimize the cross-validated risk based on negative log-likelihood loss,  \[
\epsilon_{n} = \arg\min_{\epsilon \in \mathbb{R}} \sum_{v=1}^V P_{n,v}^1 L(\cdot \mid \bar{Q}_{n,v,\epsilon}^1) \ ,
 \] 
where for a given $\bar{Q} \in \bar{\mathcal{Q}}$ and data unit $o$, \[
L(o \mid \bar{Q}) = - \{ y \log \bar{Q}(w) + (1-y) \log[1-\bar{Q}(w)]\} \ .
\] 
Let $\bar{Q}_{n,v}^{1 \#}=\bar{Q}_{n,v,\epsilon_n}^1$ be the resultant targeted OR estimator. The CV-CTMLE of $\psi_0$ is $\psi_{n,\text{cv}}^\# = \frac{1}{V} \sum_{v=1}^V P_{n,v}^1 \bar{Q}_{n,v}^{1\#}$.

\section*{Appendix E. \texttt{R} code for implementing estimators}

In this section, we provide brief snippets of \texttt{R} code for implementing the proposed estimators on simulated data. The following code simulates a sample of size $n = 100$. Here, $W$ is bivariate with a Uniform(0,1) component (\texttt{W1}) and a Bernoulli(1/2) component (\texttt{W2}). The true PS is $\bar{G}_0(w) = \mbox{expit}(w_1 - w_2)$. The true OR is $\bar{Q}_0(w) = \mbox{expit}(w_1w_2 - 1)$. The true value of the treatment-specific mean $\psi_0 \approx 0.32$. \small
\begin{lstlisting}
# sample size
n <- 100
# baseline covariates
W <- data.frame(W1 = runif(n), W2 = rbinom(n, 1, 0.5))
# true propensity score
Gbar0 <- plogis(W$W1 - W$W2)
# simulate treatment
A <- rbinom(n, 1, Gbar0)
# true outcome regression
Qbar0 <- plogis(W$W1 * W$W2 - A)
# simulate outcome
Y <- rbinom(n, 1, Qbar0)
# put all data in a data.frame
obs_dat <- data.frame(W, A = A, Y = Y)
\end{lstlisting}

\normalsize \noindent We now follow the prescribed steps to generate $\psi_n^\#$. \small 
\begin{lstlisting}
# 1. estimate OR 
# here we use correctly specified glm
est_or <- glm(Y ~ I(W1*W2), data = obs_dat[A == 1, ],
              family = binomial())

# 2. predict outcome
# on logistic scale
obs_dat$logit.Qbar_n.W <- 
    predict(est_or, newdata = obs_dat)
# and on probability scale
obs_dat$Qbar_n.W <- plogis(obs_dat$logit.Qbar_n.W)

# 3. estimate adaptive PS 
# here we use natural splines
library(splines)
est_adapt_ps <- glm(A ~ ns(Qbar_n.W, df = 2), 
                    family = binomial(),
                    data = obs_dat)

# 4. predict PS
# on probability scale
obs_dat$Gbar_n.W.Qbar_n <- 
    predict(est_adapt_ps, type = 'response') 

# 5. fit OR working model
# define covariate at observed values of A
obs_dat$H_n.AW <- A / obs_dat$Gbar_n.W.Qbar_n
# define covariate at A = 1
obs_dat$H_n.1W <- 1 / obs_dat$Gbar_n.W.Qbar_n
# fit working model
est_wm <- glm(Y ~ -1 + offset(logit.Qbar_n.W) + H_n.AW,
              family = binomial(), data = obs_dat)
# get coefficient
eps_hash_n <- est_wm$coefficients

# 6. target OR estimate
# generate prediction from working model with 
# covariate H_n.1W
Qbar_hash_n <- plogis(obs_dat$logit.Qbar_n.W + 
                        eps_hash_n * obs_dat$H_n.1W)

# 7. compute plug-in
psi_hash_n <- mean(Qbar_hash_n)
\end{lstlisting}
\normalsize

\noindent To compute the one-step estimator, we replace steps 5-7 as follows. \small

\begin{lstlisting}
# 5. evaluate influence function
Dstar.O <- obs_dat$H_n.AW * (Y - obs_dat$Qbar_n.W) + 
              obs_dat$Qbar_n.W - mean(obs_dat$Qbar_n.W)

# 6. compute one-step estimator
psi_hashos_n <- mean(obs_dat$Qbar_n.W) + mean(Dstar.O)
\end{lstlisting}
\normalsize

\section*{Appendix F. Details for Theorem 1}

\subsection*{Regularity conditions for CTMLE}
Given a $P_0$-measurable function $f_n$, we define the $L^2(P_0)$-norm of $f_n$ as $\pl f_n \pl_{P_0} := P_0(f_n^2)^{1/2}$. For each $w \in \mathcal{W}$ and given $\bar{Q}_n$ and $\bar{Q}_0$, we define \[
	\bar{G}_{0}(w \mid \bar{Q}_n, \bar{Q}_0) := \pr_{P_0}[A = 1 \mid \bar{Q}_n(W) = \bar{Q}_n(w), \bar{Q}_0(W) = \bar{Q}_0(w)] \ . 
\]
Similarly, we define \begin{multline*}
  \bar{G}_0(w \mid \bar{Q}_n, \bar{Q}_n^\#, \bar{Q}_0) := \pr_{P_0}[A = 1 \mid \bar{Q}_n(W) = \bar{Q}_n(w), \bar{Q}_n^\#(W) = \bar{Q}_n^\#(w), \\
  \bar{Q}_0(W) = \bar{Q}_0(w)] \ . 
\end{multline*}
Theorem 1 requires the following regularity conditions: \begin{enumerate}
	\item[(i)] $P_n \Dstarnadapt = o_{\text{p}}(n^{-1/2})$ ; 
	\item[(ii)] $\pl \bar{Q}_n^\# - \bar{Q}_0 \pl_{P_0} = o_{\text{p}}(n^{-1/4})$ and $\pl \bar{Q}_n - \bar{Q}_0 \pl_{P_0} = o_{\text{p}}(n^{-1/4})$; 
	\item[(iii)] $\pl \bar{G}_n(\cdot \mid \bar{Q}_n) - \bar{G}_0(\cdot \mid \bar{Q}_n) \pl_{P_0} = o_{\text{p}}(n^{-1/4})$ ; 
	\item[(iv)] $P_0[ \Dstarnadapt - \Dstarzeroadapt ]^2$ converges in-probability to zero and $\Dstarnadapt - \Dstarzeroadapt$ falls in a $P_0$-Donsker class with probability tending to one as $n$ tends to infinity;
	\item[(v)] $\bar{G}_0(\cdot \mid \bar{Q}_n, \bar{Q}_0)$ is differentiable in $\bar{Q}_0$ and \[
		\underset{n \rightarrow \infty}{\mbox{lim sup}} \left\{ \sup_{u,v} \bigg| \frac{d}{dv} \pr_{P_0}[A = 1 \mid \bar{Q}_n(W) = u, \bar{Q}_0(W) = v)\bigg| \right\} < \infty \ , 
		\] 
		where the supremum over $(u,v)$ is over the support of $(\bar{Q}_n(W),\bar{Q}_0(W))$;
  \item[(vi)] $P_0[\{\bar{G}_0(\cdot \mid \bar{Q}_n,\bar{Q}_n^\#,\bar{Q}_0) - \bar{G}_0(\cdot \mid \bar{Q}_n,\bar{Q}_0)\}/\bar{G}_0(\cdot \mid \bar{Q}_n) \ (\bar{Q}_n^\# - \bar{Q}_0)  ] = o_{\text{p}}(n^{-1/2})$. 
\end{enumerate}

\subsection*{Regularity conditions for collaborative one-step}
To prove asymptotic linearity of the collaborative one-step estimator, we require conditions (iii) and (v) above in addition to the following regularity conditions: \begin{enumerate}
  \item[(iia)] $\pl \bar{Q}_n - \bar{Q}_0 \pl_{P_0} = o_{\text{p}}(n^{-1/4})$; 
  \item[(iva)] $P_0[ D^*(\cdot \mid \bar{Q}_n, Q_{n,W}, \bar{G}_n(\cdot \mid \bar{Q}_n)) - D^*(\cdot \mid \bar{Q}_0, Q_{0,W}, \bar{G}_0(\cdot \mid \bar{Q}_0)) ]^2$ converges in-probability to zero and $D^*(\cdot \mid \bar{Q}_n, Q_{n,W}, \bar{G}_n(\cdot \mid \bar{Q}_n)) - D^*(\cdot \mid \bar{Q}_0, Q_{0,W}, \bar{G}_0(\cdot \mid \bar{Q}_0)$ falls in a $P_0$-Donsker class with probability tending to one as $n$ tends to infinity.
\end{enumerate}

\subsection*{Proof of Theorem 1}
\noindent To prove Theorem 1, we require the following lemma. 
\begin{lemma} \label{cruciallemmaforctmle}
Let $X_n$ and $X_0$ be functions of $W$ with support $\mathcal{X}_{n0} := \mathcal{X}_n \times \mathcal{X}_0$. For each $(u,v) \in \mathcal{X}_{n0}$, define $f_{0\mid n0}(u,v) := \E_{P_0}[A \mid X_n(W) = u, X_0(W) = v]$ and $f_{0\mid n}(u) := \E_{P_0}[A \mid X_n(W) = u]$. If $f_{0\mid n0}$ is differentiable in its second coordinate and \[
\underset{n \rightarrow \infty}{\mbox{lim sup}} \left[ \ \underset{(u,v) \in \mathcal{X}_{n0}}{\mbox{sup}} \bigg| \frac{d}{dv} f_{0 \mid n0}(u,v) \bigg| \ \right] < \infty \ , 
\]
then \[
	\pl f_{0 \mid n0} - f_{0 \mid n} \pl_{P_0}^2 \ \le C \pl X_n - X_0 \pl_{P_0}^2 \ . 
\]
\end{lemma}

\begin{proof}
We have that $f_{0 \mid n}$ is the projection of $f_{0 \mid n0}$ onto the subspace $L^2(X_n)$ of all functions that only depend on $X_n(W)$ endowed with the usual covariance as inner product. This is a subspace of $L^2(X_n,X_0)$, the space of all functions that depend on $X_n$ and $X_0$. Thus, 
\begin{align*}
&\pl f_{0\mid n0} - f_{0 \mid n} \pl_{P_0}^2 \\
&\hspace{0.1in} = \inf_{g \in L^2(X_n)}\pl f_{0\mid n0} - g \pl^2_{P_0}\\
&\hspace{0.1in} = \inf_{g\in L^2(X_n)} \int [f_{0\mid n0}(X_n(u), X_0(u)) - g(X_n(u))]^2 dQ_{0,W}(u)\\
&\hspace{0.1in} = \inf_{g\in L^2(X_n)} \int [f_{0\mid n0}(X_n(u), X_0(u)) - f_{0 \mid n0}(X_n(u), X_n(u)) \\
&\hspace{1.4in} + f_{0 \mid n0}(X_n(u), X_n(u)) - g(X_n(u))]^2 dQ_{0,W}(u) \\
&\hspace{0.1in} = \int [f_{0\mid n0}(X_n(u), X_0(u)) - f_{0 \mid n0}(X_n(u), X_n(u))]^2 dQ_{0,W}(u) \\
&\hspace{0.6in} + \inf_{g \in L^2(X_n)} \biggl\{ \int[f_{0 \mid n0}(X_n(u), X_n(u)) - g(X_n(u))]^2 dQ_{0,W}(w) \\
&\hspace{1.0in} + 2\int [f_{0\mid n0}(X_n(u), X_0(u)) - f_{0 \mid n0}(X_n(u), X_n(u))] \\ 
&\hspace{1.4in} \times [f_{0 \mid n0}(X_n(u), X_n(u)) - g(X_n(u))] dQ_{0,W}(w) \biggr\} \ . 
\end{align*}
The infimum is attained at $u \rightarrow f_{0 \mid n0}(X_n(u), X_n(u))$. Thus,
\[
\pl f_{0\mid n0} - f_{0 \mid n} \pl_{P_0}^2 = \int [f_{0\mid n0}(X_n(u), X_0(u)) - f_{0 \mid n0}(X_n(u), X_n(u))]^2 dQ_{0,W}(u) \ .
\]
Because $f_{0\mid n0}$ is differentiable in its second coordinate, 
\begin{align*}
f_{0\mid n0}(X_n(u), X_0(u)) &= f_{0\mid n0}(X_n(u), X_n(u)) \\
&\hspace{0.4in} + \frac{d}{dv} f_{0\mid n0}(X_n(u), v)\big|_{v=\xi(X_n(u),X_0(u))}[X_0(u)-X_n(u)],
\end{align*}
for an intermediate point $\xi(X_n(u),X_0(u))$ between $X_n(u)$ and $X_0(u)$. Due to the assumption on uniform bound on the derivative of $f_{0 \mid n0}$, \[ 
| \ f_{0\mid n0}(X_n(u), X_0(u)) - f_{0\mid n0}(X_n(u), X_n(u)) \ | \le C | \ X_n(u) - X_0(u) \ | 
\] 
for some $C<\infty$. Thus, we conclude that 
\[
\pl f_{0\mid n0} - f_{0 \mid n} \pl_{P_0}^2 \ \leq C \pl X_n - X_0 \pl_{P_0}^2 \ .
\]
\end{proof}

\noindent Theorem 1 can now be proven as follows. \begin{proof}
An exact linearization of $\psi_{n}^\#$ gives \begin{align*}
\psi_n^\# - \psi_0 \ &=\ -P_0 \Dstarnadapt + \Remstaradapt \\
&= (P_n - P_0) \Dstarnadapt - P_n \Dstarnadapt \\
&\hspace{2.4in} + \Remstaradapt \\ 
&= (P_n - P_0) \Dstarzeroadapt - P_n \Dstarnadapt \\
&\hspace{0.2in} + (P_n - P_0) [\Dstarnadapt - \Dstarzeroadapt] \\
&\hspace{2.4in} + \Remstaradapt  \ . 
\end{align*}
By assumption (i) $P_n \Dstarnadapt = o_{\text{p}}(n^{-1/2})$, while assumption (iv) is sufficient to ensure that $(P_n - P_0) [\Dstarnadapt - \Dstarzeroadapt] = o_{\text{p}}(n^{-1/2})$. Furthermore, because of the double-robustness of the EIF, $P_0 \Dstarzeroadapt = 0$. Thus, it remains to show that $\Remstaradapt = o_{\text{p}}(n^{-1/2})$. Now,
\begingroup
\allowdisplaybreaks
\begin{align*}
&\Remstaradapt \\
&\hspace{0.1in} = P_0 \left\{ \left[ \frac{\bar{G}_n(\cdot \mid \bar{Q}_n) - \bar{G}_0}{\bar{G}_n(\cdot \mid \bar{Q}_n)} \right] (\bar{Q}_n^\# - \bar{Q}_0) \right\} \\
&\hspace{0.1in} = P_0 \left\{ \left[ \frac{\bar{G}_n(\cdot \mid \bar{Q}_n) - \bar{G}_0}{\bar{G}_0(\cdot \mid \bar{Q}_n)} \right] (\bar{Q}_n^\# - \bar{Q}_0) \right\} + R_{21,n} \\
&\hspace{0.1in} = P_0 \left\{ \left[ \frac{\bar{G}_n(\cdot \mid \bar{Q}_n) - A}{\bar{G}_0(\cdot \mid \bar{Q}_n)} \right] (\bar{Q}_n^\# - \bar{Q}_0) \right\} + R_{21,n} \\
&\hspace{0.1in} = P_0 \left\{ \left[ \frac{\bar{G}_n(\cdot \mid \bar{Q}_n) - \bar{G}_0(\cdot \mid \bar{Q}_n, \bar{Q}_0)}{\bar{G}_0(\cdot \mid \bar{Q}_n)} \right] (\bar{Q}_n^\# - \bar{Q}_0) \right\} + R_{21,n} + R_{22,n}\\ 
&\hspace{0.1in} = P_0 \left\{ \left[ \frac{\bar{G}_n(\cdot \mid \bar{Q}_n) - \bar{G}_0(\cdot \mid \bar{Q}_n) + \bar{G}_0(\cdot \mid \bar{Q}_n) - \bar{G}_0(\cdot \mid \bar{Q}_n, \bar{Q}_0)}{\bar{G}_0(\cdot \mid \bar{Q}_n)} \right] (\bar{Q}_n^\# - \bar{Q}_0) \right\} \\
&\hspace{0.8in} + R_{21,n} + R_{22,n}\\ 
&\hspace{0.1in} = R_{23,n} + R_{24,n} + R_{21,n} + R_{22,n}\ , 
\end{align*}
\endgroup
where \begin{align*}
R_{21,n} &:= P_0 \left\{ \frac{\bar{G}_n(\cdot \mid \bar{Q}_n) - \bar{G}_0}{\bar{G}_0(\cdot \mid \bar{Q}_n) \bar{G}_n(\cdot \mid \bar{Q}_n)} [\bar{G}_0(\cdot \mid \bar{Q}_n) - \bar{G}_n(\cdot \mid \bar{Q}_n)] [\bar{Q}_n^\# - \bar{Q}_0] \right\} \\
R_{22,n} &:= P_0\left[ \left\{\frac{\bar{G}_0(\cdot \mid \bar{Q}_n,\bar{Q}_n^\#,\bar{Q}_0) - \bar{G}_0(\cdot \mid \bar{Q}_n,\bar{Q}_0}{\bar{G}_0(\cdot \mid \bar{Q}_n)} \right\} (\bar{Q}_n^\# - \bar{Q}_0)  \right] \\
R_{23,n} &:= P_0 \left\{ \left[ \frac{\bar{G}_n(\cdot \mid \bar{Q}_n) - \bar{G}_0(\cdot \mid \bar{Q}_n)}{\bar{G}_0(\cdot \mid \bar{Q}_n)} \right] (\bar{Q}_n^\# - \bar{Q}_0) \right\} \\ 
R_{24,n} &:= P_0 \left\{ \left[ \frac{\bar{G}_0(\cdot \mid \bar{Q}_n) - \bar{G}_0(\cdot \mid \bar{Q}_n, \bar{Q}_0)}{\bar{G}_0(\cdot \mid \bar{Q}_n)} \right] (\bar{Q}_n^\# - \bar{Q}_0) \right\} \ . 
\end{align*}
Assumptions (ii) and (iii) are sufficient to ensure that $R_{21,n} = o_{\text{p}}(n^{-1/2})$ and that $R_{23,n} = o_{\text{p}}(n^{-1/2})$. By assumption (vi), $R_{22,n} = o_{\text{p}}(n^{-1/2})$. For $R_{24,n}$, assumption (v) allows us to apply Lemma \ref{cruciallemmaforctmle} with $f_{0 \mid n0} = \bar{G}_0(\cdot \mid \bar{Q}_n, \bar{Q}_0)$ and $f_{0 \mid n} = \bar{G}_0(\cdot \mid \bar{Q}_n)$. Thus, assumption (ii) allows us to conclude that $R_{24,n} = o_{\text{p}}(n^{-1/2})$. 
\end{proof}

\subsection*{Proof of asymptotic linearity of one-step estimator}

The proof of asymptotic linearity for the one-step estimator is similar to that for the CTMLE, but is rather more straightforward as we do not need conditions accounting for the targeted fit. \begin{proof}
A linearization of the plug-in estimator based on $\bar{Q}_n$ combined with the one-step correction yields \begin{align*}
\psi_{n,\text{os}}^\# - \psi_0 \ &=\ (P_n-P_0) \Dstarnadaptinit \\
&\hspace{0.8in} + \Remstaradaptinit \\
&= (P_n - P_0) \Dstarzeroadapt \\
&\hspace{0.2in} + (P_n - P_0) [\Dstarnadaptinit - \Dstarzeroadapt] \\
&\hspace{2.4in} + \Remstaradaptinit  \ . 
\end{align*}
Assumption (iva) is sufficient to ensure that $(P_n - P_0) [\Dstarnadaptinit - \Dstarzeroadapt] = o_{\text{p}}(n^{-1/2})$. Furthermore, because of the double-robustness of the EIF, $P_0 \Dstarzeroadapt = 0$. Thus, it remains to show that $\Remstaradaptinit = o_{\text{p}}(n^{-1/2})$. \begin{align*}
&\Remstaradaptinit \\
&\hspace{0.1in} = P_0 \left\{ \left[ \frac{\bar{G}_n(\cdot \mid \bar{Q}_n) - \bar{G}_0}{\bar{G}_n(\cdot \mid \bar{Q}_n)} \right] (\bar{Q}_n - \bar{Q}_0) \right\} \\
&\hspace{0.1in} = P_0 \left\{ \left[ \frac{\bar{G}_n(\cdot \mid \bar{Q}_n) - \bar{G}_0}{\bar{G}_0(\cdot \mid \bar{Q}_n)} \right] (\bar{Q}_n - \bar{Q}_0) \right\} + R_{21,n,\text{os}} \\
&\hspace{0.1in} = P_0 \left\{ \left[ \frac{\bar{G}_n(\cdot \mid \bar{Q}_n) - A}{\bar{G}_0(\cdot \mid \bar{Q}_n)} \right] (\bar{Q}_n - \bar{Q}_0) \right\} + R_{21,n,\text{os}} \\
&\hspace{0.1in} = P_0 \left\{ \left[ \frac{\bar{G}_n(\cdot \mid \bar{Q}_n) - \bar{G}_0(\cdot \mid \bar{Q}_n, \bar{Q}_0)}{\bar{G}_0(\cdot \mid \bar{Q}_n)} \right] (\bar{Q}_n - \bar{Q}_0) \right\} + R_{21,n,\text{os}} \\ 
&\hspace{0.1in} = P_0 \left\{ \left[ \frac{\bar{G}_n(\cdot \mid \bar{Q}_n) - \bar{G}_0(\cdot \mid \bar{Q}_n) + \bar{G}_0(\cdot \mid \bar{Q}_n) - \bar{G}_0(\cdot \mid \bar{Q}_n, \bar{Q}_0)}{\bar{G}_0(\cdot \mid \bar{Q}_n)} \right] (\bar{Q}_n - \bar{Q}_0) \right\} \\
&\hspace{0.8in} + R_{21,n,\text{os}} \\ 
&\hspace{0.1in} = R_{22,n,\text{os}} + R_{23,n,\text{os}} + R_{21,n,\text{os}}\ , 
\end{align*}
where
\begin{align*}
R_{21,n,\text{os}} &:= P_0 \left\{ \frac{\bar{G}_n(\cdot \mid \bar{Q}_n) - \bar{G}_0}{\bar{G}_0(\cdot \mid \bar{Q}_n) \bar{G}_n(\cdot \mid \bar{Q}_n)} [\bar{G}_0(\cdot \mid \bar{Q}_n) - \bar{G}_n(\cdot \mid \bar{Q}_n)] [\bar{Q}_n - \bar{Q}_0] \right\} \\
R_{22,n,\text{os}} &:= P_0 \left\{ \left[ \frac{\bar{G}_n(\cdot \mid \bar{Q}_n) - \bar{G}_0(\cdot \mid \bar{Q}_n)}{\bar{G}_0(\cdot \mid \bar{Q}_n)} \right] (\bar{Q}_n - \bar{Q}_0) \right\} \\ 
R_{23,n,\text{os}} &:= P_0 \left\{ \left[ \frac{\bar{G}_0(\cdot \mid \bar{Q}_n) - \bar{G}_0(\cdot \mid \bar{Q}_n, \bar{Q}_0)}{\bar{G}_0(\cdot \mid \bar{Q}_n)} \right] (\bar{Q}_n - \bar{Q}_0) \right\} \ . 
\end{align*}
Assumptions (ii) and (iii) are sufficient to ensure that $R_{21,n,\text{os}} = o_{\text{p}}(n^{-1/2})$ and that $R_{22,n,\text{os}} = o_{\text{p}}(n^{-1/2})$. For $R_{23,n,\text{os}}$, assumption (v) allows us to apply Lemma \ref{cruciallemmaforctmle} with $f_{0 \mid n0} = \bar{G}_0(\cdot \mid \bar{Q}_n, \bar{Q}_0)$ and $f_{0 \mid n} = \bar{G}_0(\cdot \mid \bar{Q}_n)$. Thus, assumption (ii) allows us to conclude that $R_{23,n,\text{os}} = o_{\text{p}}(n^{-1/2})$. 
\end{proof}

\subsection*{Discussion of regularity conditions}

Regularity condition (i) stipulates that the EIF estimating equation is approximately solved. While this can generally be ensured by utilizing a TMLE based on a universally least favorable parametric submodel \citep{vanderLaan&Gruber15}, we have instead proposed a TMLE based on a locally least favorable submodel. Nevertheless, the theorem in the Appendix of \citet{vanderLaan15} suggests that because (by assumption) $\bar{Q}_n$ and $\bar{G}_n(\cdot \mid \bar{Q}_n)$ are $n^{-1/4}$-consistent estimators of $\bar{Q}_0$ and $\bar{G}_0(\cdot \mid \bar{Q}_n),$ respectively, then we can expect regularity condition (i) to be satisfied for the proposed TMLE. 

Regularity condition (ii) stipulates that the $L^2(P_0)$-norm of both the initial $\bar{Q}_n - \bar{Q}_0$ and targeted $\bar{Q}_n^\# - \bar{Q}_0$ are $o_{\text{p}}(n^{-1/4})$. The former condition is generally sufficient to ensure the latter \citep{vanderLaan&Gruber15}. Regularity condition (iii) similarly stipulates that given $\bar{Q}_n$ our estimator of $\bar{G}_0(\cdot \mid \bar{Q}_n)$ converges with respect to $L^2(P_0)$-norm faster than $n^{-1/4}$. These assumptions can be satisfied under weak conditions by using the highly adaptive lasso minimum loss estimator (HAL-MLE) to estimate $\bar{Q}_0$ and $\bar{G}_0(\cdot \mid \bar{Q}_n)$ \citep{vanderLaan15,Benkeser&vanderLaan16}. Alternatively, the highly adaptive lasso can be included as a candidate in a cross-validation-based estimator selection routine or ensemble routine (i.e., a super learner \citet{vanderLaan:Polley:Hubbard07}), while maintaining the relevant convergence properties \citep{vanderLaan15}. 

Regularity condition (iv) is an empirical process condition that places restrictions on the adaptivity of the initial estimators. As above, this condition can be satisfied under weak conditions by selecting the HAL-MLE as initial estimator. Alternatively, if all candidate estimators in a super learner routine themselves satisfy this condition, then the resultant ensemble estimator will also satisfy this condition. On the other hand, this assumption can be wholly obviated by using a CVTMLE routine described in Appendix F. 

Regularity condition (v) is a technical condition needed to bound the convergence rate of $\bar{G}_0(\cdot \mid \bar{Q}_n, \bar{Q}_0) - \bar{G}_n(\cdot \mid \bar{Q}_n)$ by the convergence rate of $\bar{Q}_n - \bar{Q}_0$. It is essentially a smoothness assumption on the parameter $\bar{G}_0(\cdot \mid \bar{Q}_n, \bar{Q}_0)$. While this condition (along with condition (ii)) is sufficient to ensure that $R_{23,n} = o_{\text{p}}(n^{-1/2})$, it is not necessary. Therefore, there may be weaker conditions that would imply the negligibility of $R_{23,n}$. We leave exploration of these conditions to future work. Regularity condition (vi) is another technical condition. Essentially this assumption is only needed because we based our adaptive propensity score estimate on the initial estimator $\bar{Q}_n$ rather than a targeted estimator $\bar{Q}_n^\#$, thereby avoiding the need for an iterative TMLE, which can become computationally expensive. Nevertheless, the resultant computational expediency comes at the price of this assumption. We suspect that Lemma 1 could be extended to provided conditions under which the negligibility of this term could be ensured. Alternatively, one could implement an iterative TMLE wherein one alternates between targeting the current OR estimate based on the current adaptive PS estimate and updating the adaptive PS estimate by regressing the treatment on the current OR estimate. Finally, as a third alternative, a TMLE based on a universal least favorable parametric submodel could be implemented \citep{vanderLaan&Gruber15}. However, this would require even greater computational burden as we would likely be required to perform more adaptive PS estimation based on an incrementally updated outcome regression. In our experience, the computationally expedient versions of TMLE generally perform as well as those that allow for weaker assumptions (but are in turn more computationally expensive). 

\section*{Appendix G. Results for one-step estimators}

A comparison of the collaborative vs. standard one-step estimators for Simulation 1 arrives at essentially the same conclusions as for CTMLE vs. TMLE (Figures \ref{onestep_estimator_results_sim1} and \ref{onestep_coverage_results_sim1}). 

\begin{figure}[ht]
\centering
\includegraphics[width = \textwidth]{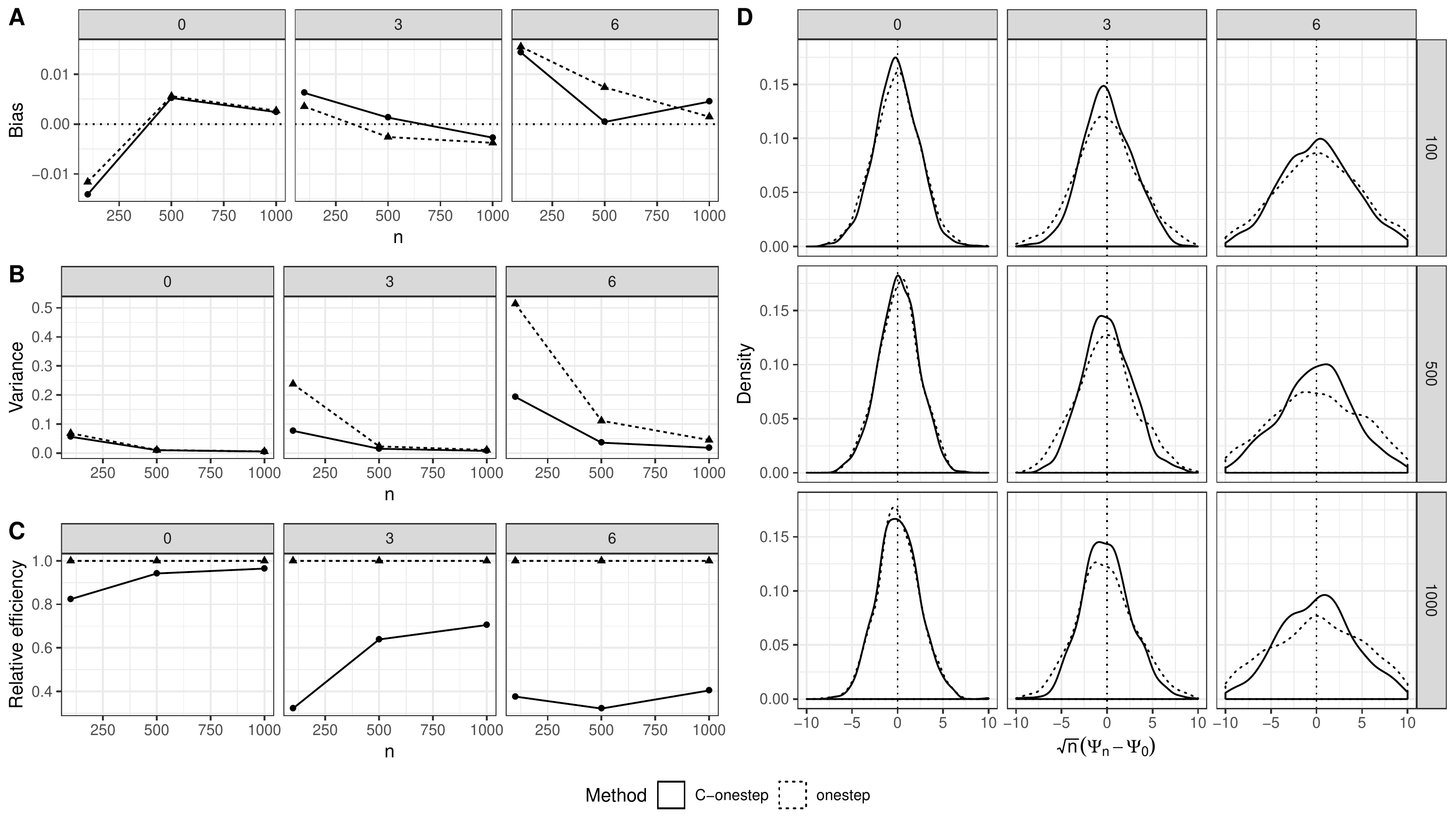}
\caption{Results for simulation 1 comparing collaborative one-step and standard one-step. Each panel displays a different performance metric and each sub-panel displays results for $\gamma \in \{0,3, 6\}$, representing, respectively, settings with no positivity, moderate positivity, and extreme positivity violations. Panel A: Bias of the estimators. Panel B: Variance of the estimators. Panel C: Relative efficiency (defined as ratio of mean squared-error) of collaborative vs. standard one-step. Numbers below one indicate greater efficiency of collaborative one-step. Panel D: Kernel density estimates of sampling distributions using a Gaussian kernel and Silverman's rule of thumb bandwidth \citep{silverman1986density}.}
\label{onestep_estimator_results_sim1}
\end{figure}

\begin{figure}[ht]
\centering
\includegraphics[width = 0.5\textwidth]{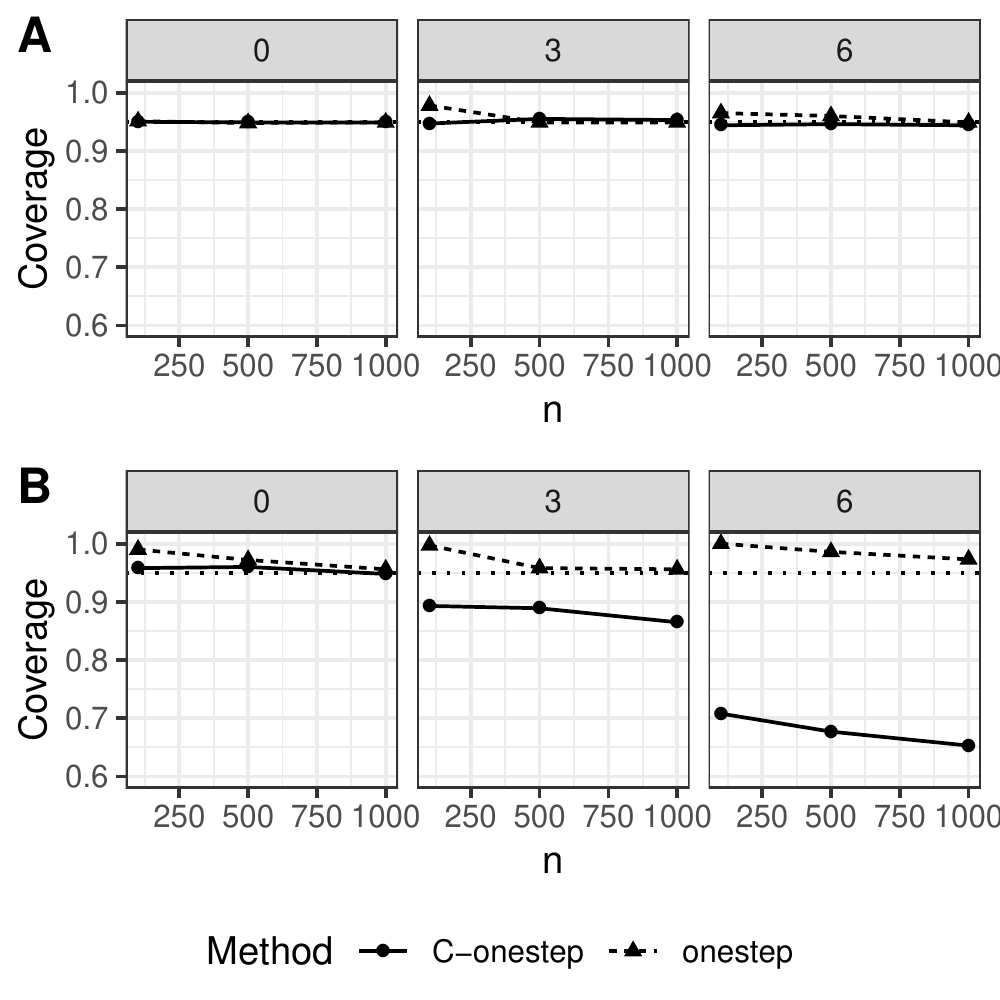}
\caption{Results for simulation 1 comparing confidence intervals for collaborative one-step and standard one-step. Each panel displays the coverage as a function of sample size and each sub-panel displays results for $\gamma \in \{0,3, 6\}$. Panel A: Coverage probability of nominal 95\% oracle confidence intervals. Panel B: Coverage probability of nominal 95\% confidence intervals based on estimated standard errors.}
\label{onestep_coverage_results_sim1}
\end{figure}

\section*{Appendix H. Details of data analysis}

\subsection*{Missing covariate data}

There were several variables with small amounts of missing data: paternal age (70 missing), paternal education (118), maternal age at first child birth (52), and household income (13). For each of these variables, we imputed missing values using super learning. First, we fit a super learner (based on 10-fold cross-validation) that used a generalized linear model, a random forest, and an intercept only model using all participants' data who had the variable measured. We then filled in the missing values with the predicted value from this super learner. Finally, we created an additional indicator variable that indicated whether or not the value in the corresponding variable was imputed or actually observed. While this approach may not be the ideal approach to handling missing covariate data (see, for example, the discussion in \citet{groenwold2012missing}), we feel that the amount of missingness is low enough that bias incurred by covariates missing not at random is likely to be small. 

\subsection*{Super learner estimation of multi-level treatment}

We took a simple approach to generating an estimate of the propensity for pre-, full-, and post-term births. First, we estimated the probability of a post-term birth by fitting a super learner with outcome $\ind(A = \mbox{post-term})$. For each $w \in \mathcal{W}$, we denote by $\bar{G}_{0,\text{post}}(w) := \mbox{pr}_{P_0}(A = \text{post-term} \mid W = w)$. Then we fit a super learner with outcome $\ind(A = \text{pre-term})$ but restricted to the subset of participants with $A \ne \text{post-term}$. This super learner provides an estimate of $\tilde{G}_{0,\text{pre}}(w) := \mbox{pr}_{P_0}(A = \mbox{pre-term} \mid W = w, A \ne \text{post-term})$, which can be mapped into an estimate $\bar{G}_{n,\text{pre}}$ of the propensity for pre-term birth $\bar{G}_{0,\text{pre}}$: \[
  \bar{G}_{n,\text{pre}}(w) := \tilde{G}_{n,\text{pre}}(w) [1 - \bar{G}_{n,\text{post}}(w)] \ . 
\]
Finally, we let the estimated propensity for full-term birth $\bar{G}_{n,\text{full}} = 1 - \bar{G}_{n,\text{post}} - \bar{G}_{n,\text{pre}}$. 

\subsection*{Super learner results}

The super learner for the outcome regression identified lasso as the best single algorithm, but assigned ensemble weights to several algorithms (Table \ref{sl_or_table}). Overall, the super learner ensemble was found to outperform any single algorithm and the cross-validation-selected algorithm (i.e., the discrete super learner) (Table \ref{sl_or_table}). 

\begin{table}[ht]
\centering
\begin{tabular}{lcccc}
  \hline
Algorithm & Avg. risk & SE risk & Range risk & SL coef. \\ 
  \hline
  \texttt{Super Learner} & 0.767 & 0.021 & (0.698, 0.831) & -- \\ 
  \texttt{Discrete SL} & 0.781 & 0.022 & (0.695, 0.865) & -- \\ 
  \hline
  \texttt{SL.glm} & 0.786 & 0.022 & (0.715, 0.860) & 0.000 \\ 
  \texttt{SL.earth} & 0.837 & 0.024 & (0.766, 0.942) & 0.214 \\ 
  \texttt{SL.ranger} & 0.794 & 0.021 & (0.701, 0.873) & 0.109 \\ 
  \texttt{SL.glmnet} & 0.777 & 0.021 & (0.693, 0.865) & 0.000 \\ 
  \texttt{SL.gbm} & 0.784 & 0.022 & (0.695, 0.880) & 0.145 \\ 
  \texttt{SL.mean} & 1.010 & 0.026 & (0.918, 1.181) & 0.007 \\ 
  \texttt{SL.step.forward} & 0.777 & 0.022 & (0.696, 0.849) & 0.525 \\ 
   \hline
\end{tabular}
\caption{Super learner results for the outcome regression. Avg. risk column is the estimated cross-validated mean squared-error of the algorithms. SE risk is the estimated standard error of the cross-validated risk estimates. Range risk is the range of the risk across the five-folds of cross-validation. SL coef. is the super learner ensemble weight assigned to each candidate algorithm.}
\label{sl_or_table}
\end{table}

\begin{table}[ht]
\centering
\begin{tabular}{lccc}
  \hline
Algorithm & Avg. risk & Range risk & SL coef. \\ 
  \hline
  \texttt{Super Learner} & 0.264 & (0.219, 0.317) & -- \\ 
  \texttt{Discrete SL} & 0.264 & (0.219, 0.316) & -- \\ 
  \hline
  \texttt{SL.glm} & 0.286 & (0.248, 0.330) & 0.011 \\ 
  \texttt{SL.earth} & 0.311 & (0.232, 0.378) & 0.000 \\ 
  \texttt{SL.ranger} & 0.282 & (0.228, 0.342) & 0.000 \\ 
  \texttt{SL.glmnet} & 0.264 & (0.219, 0.315) & 0.000 \\ 
  \texttt{SL.gbm} & 0.264 & (0.219, 0.316) & 0.013 \\ 
  \texttt{SL.mean} & 0.263 & (0.219, 0.314) & 0.723 \\ 
  \texttt{SL.step.forward} & 0.265 & (0.222, 0.325) & 0.254 \\ 
   \hline
\end{tabular}
\caption{Super learner results for the first propensity score model (predicting post-term birth vs. other). Avg. risk column is the estimated cross-validated mean squared-error of the algorithms. Range risk is the range of the risk across the five-folds of cross-validation. SL coef. is the super learner ensemble weight assigned to each candidate algorithm.}
\label{sl_ps1_table}
\end{table}

\begin{table}[ht]
\centering
\begin{tabular}{lccc}
  \hline
Algorithm & Avg. risk & Range risk & SL coef. \\ 
  \hline
\texttt{Super Learner} & 0.440 & (0.371, 0.482) & -- \\ 
  \texttt{Discrete SL} & 0.438 & (0.372, 0.478) & -- \\
  \hline 
  \texttt{SL.glm} & 0.463 & (0.370, 0.553) & 0.011 \\ 
  \texttt{SL.earth} & 0.472 & (0.439, 0.523) & 0.000 \\ 
  \texttt{SL.ranger} & 0.458 & (0.393, 0.497) & 0.000 \\ 
  \texttt{SL.glmnet} & 0.438 & (0.370, 0.479) & 0.000 \\ 
  \texttt{SL.gbm} & 0.437 & (0.370, 0.478) & 0.013 \\ 
  \texttt{SL.mean} & 0.439 & (0.372, 0.482) & 0.723 \\ 
  \texttt{SL.step.forward} & 0.447 & (0.369, 0.482) & 0.254 \\ 
   \hline
\end{tabular}
\caption{Super learner results for the second propensity score model (predicting pre-term birth vs. full-term birth). Avg. risk column is the estimated cross-validated mean squared-error of the algorithms. Range risk is the range of the risk across the five-folds of cross-validation. SL coef. is the super learner ensemble weight assigned to each candidate algorithm.}
\label{sl_ps2_table}
\end{table}

\end{document}